\documentclass{article} 

\usepackage{color}
\newcommand{\ynote}[1]{\ifnotes $\ll$\textsf{\color{red} Yael: { #1}}$\gg$ \fi}

\usepackage{amssymb}
\usepackage{amsmath}
\usepackage{amsthm}
\usepackage{fancyhdr}
\usepackage{graphics}
\usepackage{graphicx}

\usepackage{authblk}

\usepackage{environ}
\usepackage{framed}
\usepackage{url}
\usepackage[margin=1.0in]{geometry}

\newcommand{\DTIME}{\mathrm{DTIME}}
\newcommand{\cc}{\mathrm{cc}}
\newcommand{\Sim}{\mathrm{Sim}}
\newcommand{\negl}{\mathrm{negl}}
\newcommand{\poly}{\text{poly}}
\newcommand{\polylog}{\text{polylog}}
\newcommand{\E}{\mathbb{E}}

\newcommand{\GOOD}{\text{GOOD}}

\newcommand{\PSPACE}{\mathrm{PSPACE}}
\newcommand{\SPACE}{\mathrm{SPACE}}

\newcommand{\IP}{\mathrm{IP}}
\newcommand{\MIP}{\mathrm{MIP}}
\newcommand{\NEXP}{\mathrm{NEXP}}
\newcommand{\EXP}{\mathrm{EXP}}

\newcommand{\Pt}{\mathrm{P}}

\newcommand{\NS}{\mathrm{NS}}
\newcommand{\subNS}{\mathrm{subNS}}

\newcommand{\hrNS}{\mathrm{hrNS}}

\theoremstyle{plain}
\newtheorem{lemma}{Lemma}[section]
\newtheorem*{lemma*}{Lemma}
\newtheorem{theorem}[lemma]{Theorem}
\newtheorem{corollary}[lemma]{Corollary}
\newtheorem*{corollary*}{Corollary}

\newtheorem{claim}[lemma]{Claim}

\theoremstyle{definition}
\newtheorem{definition}[lemma]{Definition}

\begin{document}

\title{Non-Signaling Proofs with $O(\sqrt{\log n})$ Provers are in $\PSPACE$}
\author[1]{Dhiraj Holden\thanks{dholden@mit.edu}}
\author[2]{Yael Kalai\thanks{yael@microsoft.com}}
\affil[1]{MIT}
\affil[2]{Microsoft and MIT}
\maketitle
\begin{abstract}
Non-signaling proofs, motivated by quantum computation, have found  applications in cryptography and hardness of approximation.  
An important open problem is characterizing the power of non-signaling proofs. It is known that non-signaling proofs with two provers are characterized by $\PSPACE$ and that non-signaling proofs with $\poly(n)$-provers are characterized by $\EXP$.
However, the power of $k$-prover non-signaling proofs, for $2<k<\poly(n)$ remained an open problem.

We show that $k$-prover non-signaling proofs  (with negligible soundness) for $k=O(\sqrt{\log n})$ are contained in $\PSPACE$.   
We prove this via two different routes that are of independent interest. In both routes we consider a relaxation of non-signaling called sub-non-signaling. Our main technical contribution (which is used in both our proofs) is a reduction showing how to convert any  sub-non-signaling strategy with value at least $1-2^{-\Omega(k^2)}$ into a non-signaling one with value at least $2^{-O(k^2)}$.      

In the first route, we show that the classical prover reduction method for converting $k$-prover games into $2$-prover games carries over to the non-signaling setting with the following loss in soundness: if a $k$-prover game has value less than $2^{-ck^2}$ (for some constant~$c>0$), then the corresponding 2-prover game has value less than $1 - 2^{dk^2}$  (for some constant~$d>0$). 
In the second route we show that the value of a sub-non-signaling game can be approximated in space that is polynomial in the communication complexity and exponential in the number of provers.

\end{abstract}


\section{Introduction}
Proofs lie at the heart of the theory of computation.  
In the mid-eighties, the seminal work of Goldwasser, Micali and Rackoff~\cite{GoldwasserMR85} introduced the idea of using randomness and interaction in proofs. Interactive proofs ($\IP$) were introduced for the purpose of constructing zero-knowledge proofs, though were realized to be quite powerful in Shamir's celebrated $\IP=\PSPACE$ Theorem~\cite{LundFKN90, Shamir90}.

Shortly after interactive proofs were introduced, multi-prover interactive proofs were introduced by Ben-Or, Goldwasser, Kilian and Wigderson~\cite{Ben-OrGKW89}.  In a multi-prover interactive proof ($\MIP$) a verifier is interacting with several non-communicating provers. This class was proven to be extremely powerful, by Babai, Fortnow and Lund, who showed that  $\MIP=\NEXP$~\cite{BabaiFL91}. 
The power of this class stems from the assumption that the provers behave locally, namely that they see only the messages sent to them and do not have any information about messages sent to the other provers.  

In reality, however, it is not clear how to ensure that the provers behave locally.  Even if the provers are placed in different rooms with no communication channels between them, they may share quantum entanglement, which can cause their strategies to be correlated and non-local.
These attacks can be powerful even though at first they may seem to be benign~\cite{CleveHTW04}. 

These quantum strategies motivated the notion of non-signaling strategies, which is the subject of this work.  The notion of non-signaling strategies was first studied in physics in the context of Bell inequalities by Khalfin
and Tsirelson \cite{KT85} and Rastall \cite{Ras85}, and it has gained much attention after it was reintroduced by Popescu
and Rohrlich \cite{PR94}.
Non-signaling attacks are more general than quantum attacks;  in a non-signaling attack the cheating provers can collude, and thus each answer can be a function of {\em all} the queries.  The only restriction is that for any subset of provers, the answers provided by these provers should not convey any information about the queries given to the other provers.  Namely, the only restriction that is placed on the (possibly colluding) cheating provers is that their answers cannot be seen as ``evidence" that information has travelled between them. 


Importantly,  although non-signaling strategies are motivated by quantum
 entanglement, they found compelling applications outside the realm of quantum physics.  In particular, they have been proved to be instrumental for constructing succinct delegation schemes (under standard cryptographic assumptions) and in the realm of hardness of approximation.  

\paragraph{The applicability of non-signaling to computation delegation.}

Kalai, Raz, and Rothblum~\cite{KalaiRR13} demonstrated the significance of non-signaling by showing that any $\MIP$ that is secure against non-signaling attacks\footnote{To be precise, \cite{KalaiRR13} considered a slightly more relaxed notion, which they called statistical non-signaling.  We neglect this difference here.} can be converted into a single-prover one-round proof system (with computational soundness).  More specifically, they show that the PIR (or FHE) heuristic, proposed by Biehl, Meyer, and Wetzel \cite{BiehlMW98}, for converting any MIP to a {\em single-prover} one-round proof system  is sound if the underlying MIP has non-signaling soundness.  

In~\cite{KalaiRR14}, the same authors constructed an $\MIP$ that is secure against non-signaling attacks for every language in~$\EXP$, thus yielding the first one-round delegation scheme for {\em all} deterministic computations, under standard cryptographic assumptions.  This application of non-signaling to computation delegation has proved to be very fruitful, and yielded numerous followup works~(e.g., \cite{KalaiP16,BrakerskiHK17,BadrinarayananK18, KalaiPY19}).  Moreover, all one-round delegation schemes in the literature that are based on standard cryptographic assumptions use the concept of non-signaling.

\paragraph{The applicability of non-signaling to hardness of approximation.}

Kalai, Raz and Regev~\cite{KalaiRR16} showed the significance of non-signaling to hardness of approximation. In particular, they showed that it is hard to approximate the value of a linear program in space $2^{{\log n}^{o(1)}}$, even if the polytope is {\em fixed} (i.e., even if the algorithm has unbounded time to preprocess the polytope), and even if all the coefficients are non-negative (which is the regime where hardness of approximation is most meaningful).  More specifically, they showed that there exists a fixed polytope (corresponding to the set of all possible non-signaling strategies) such that approximating the value of a linear program (where the coefficients of the objective function and the variables are restricted to be positive) is $\Pt$-complete with a $\polylog$-space reduction.  Prior work~\cite{DLR79,Serna91,FeigeK97} demonstrated such hardness of approximation for the case where the polytope was not fixed  (and preprocessing is not allowed).

\noindent The importance of the notion of non-signaling gives rise to the  following fundamental question: 

\medskip {\em What is the power of  multi-prover interactive proofs that are sound against non-signaling strategies? } 

\medskip \noindent This is precisely the question we study in this work.  In what follows, we denote the class of one-round multi-prover interactive proofs with non-signaling soundness by $\NS$ $\MIP$.  We denote by $k$-prover $\NS$ $\MIP$ the class of one-round $k$-prover interactive proofs with non-signaling soundness.

\subsection{Prior Work} 
Ito,  Kobayashi and Matsumoto~\cite{IKM09} proved that  2-prover $\NS$ $\MIP$ contains $\PSPACE$ (by proving that the 2-prover scheme of Cai, Condon, and Lipton~\cite{CaiCL94} is in fact secure against non-signaling strategies).  Shortly after, Ito~\cite{Ito10} proved that  2-prover $\NS$ $\MIP$ is contained in $\PSPACE$, thus characterizing the power of 
  2-prover $\NS$ $\MIP$.
The power of 
 $k$-prover $\NS$ $\MIP$,  for $k> 2$, remained open.  

It is known that $\NS$ $\MIP$ is contained in $\EXP$ (\cite{IKM09}, implicit in \cite{DLNNR04}) since one can find the best non-signaling strategy by solving an exponential-size linear program. Therefore, the power of a $k$-prover $\NS$ $\MIP$ lies between $\PSPACE$ and $\EXP$.
More recently, Kalai, Raz and Rothblum~\cite{KalaiRR14} showed that there exists a $\poly(n)$-prover $\NS$ $\MIP$ for~$\EXP$, thus characterizing the power of $k$-prover $\NS$ $\MIP$ for $k=\poly(n)$.\footnote{More specifically, it was shown in~\cite{KalaiRR14} that there exists a constant $c\in\mathbb{N}$ such that there exists a $(\log T)^c$-prover $\NS$ $\MIP$ for $\DTIME(T)$.}

These works left open the following question:  What is the power of $k$-prover $\NS$ $\MIP$ for $2<k<\poly(n)$?
This question was studied by Chiesa, Manohar and Shinkar  in~\cite{ChiesaMS19}, who constructed a  $k$-prover $\NS$ $\MIP$ for $\EXP$ with $k=O(1)$, albeit where the verifier's queries are of exponential length.\footnote{Using the terminology of \cite{ChiesaMS19}, they construct an exponential size  no-signaling PCP for $\EXP$ with constant number of queries.}

\subsection{Our Results}\label{sec:intro:results}
Throughout this manuscript, we  assume that an $\MIP$ has completeness at least $1-\negl(n)$, and has soundness $\negl(n)$, for some negligible function $\negl(n)$.\footnote{A function $\mu:\mathbb{N}\rightarrow\mathbb{N}$ is said to be negligible if approaches zero faster than the inverse of any polynomial.}  This assumption is standard in cryptography.  We mention that often in the definition of  interactive proofs,  completeness is required to be greater than $2/3$ and soundness at most $1/3$;  this is because it is well known that this gap can be amplified to $1-\negl(n)$ and $\negl(n)$ via parallel repetition, at least for the case of single prover interactive proofs.  A parallel repetition theorem is also known for  2-prover $\MIP$s; this was proven in the classical setting by Raz~\cite{Raz95}, and in the non-signaling setting by Holenstein~\cite{H07}. In the multi-prover regime, where the number of provers is greater than 2, we do not have a parallel repetition theorem.  Moreover, in the non-signaling setting, Holmgren and Yang~\cite{HY19} provided a negative result, demonstrating that (in general) soundness cannot be amplified via parallel repetition.

We prove that $k$-prover $\NS$ $\MIP$ with $k=O(\sqrt{\log n})$ is contained in $\PSPACE$.  More generally, we prove the following theorem.   

\begin{theorem}[Informal]\label{thm:intro:main}
There exist constants $c,d>0$ such that any $k$-prover $\MIP$ with non-signaling soundness at most $2^{-ck^2}$ and completeness at least $1-2^{-dk^2}$, is contained in $\SPACE\left(\poly(n,2^{k^2})\right)$.
\end{theorem}

We emphasize that this theorem holds only for $\MIP$s that have negligible soundness and almost perfect completeness. In particular, we don't rule out the existence of a $3$-prover $\MIP$ with NS soundness $1/3$ and completeness $2/3$ for $\EXP$.  However, the soundness and completeness gap of such $\MIP$s could not be amplified (to $1-\negl(n)$) without adding provers.

We present two alternative routes for proving   Theorem~\ref{thm:intro:main}, each is of independent interest.  Both routes consider the more relaxed notion of {\em sub-non-signaling}, as  defined in~\cite{LancienW15} (for the goal of obtaining a parallel repetition theorem for non-signaling strategies).  
Both rely on the following theorem that asserts that  one can convert any sub-non-signaling strategy into a non-signaling one, albeit with a substantial loss in the success probability.  

In the following theorem we think of the input~$x$ as being fixed.  Usually, when the input is fixed, the $\MIP$ is referred to as a game. 

\begin{theorem}[Informal]\label{thm:intro:subNS-to-NS}
There exist constants $c,d>0$ such that for any $k$-prover game, if there exists a sub-non-signaling strategy that convinces the verifier to accept with probability at least $1-2^{-ck^2}$ then there exists a non-signaling strategy that  convinces the verifier to accept with probability at least $2^{-dk^2}$.
\end{theorem}
The proof of this theorem contains the bulk of technical difficulty of this work, and is used as a building block in both proofs of Theorem~\ref{thm:intro:main}.  We defer the proof overview of Theorem~\ref{thm:intro:subNS-to-NS} to Section~\ref{sec:intro:overview:proof:sub-to-ns}, and the formal proof to Section~\ref{sec:proof:main}.

We note that a related theorem was proven by Lancien and Winter~\cite{LancienW15}, who showed that for every game with full support, if there exists a sub-non-signaling strategy that succeeds with probability at least $1-\epsilon$ then there exists a  non-signaling strategy that succeeds with probability at least $1-\Gamma\epsilon$, where $\Gamma$ may be as large as exponential in the communication complexity.  This bound does not seem to be tight enough in order to obtain Theorem~\ref{thm:intro:main}.

We next present our two alternative routes for proving Theorem~\ref{thm:intro:main} (using Theorem~\ref{thm:intro:subNS-to-NS}).  The first is via a prover reduction method, and the second is via approximating the sub-non-signaling value efficiently.

\paragraph{Reducing the number of provers.}
We show that (a slight variant of) the classical prover reduction method for converting a $k$-prover $\MIP$ into a $2$-prover $\MIP$ carries over to the non-signaling setting, albeit with a substantial loss in soundness (which depends on~$k$).  

More specifically, in the seminal work of Ben-Or, Goldwasser, Kilian and Wigderson~\cite{Ben-OrGKW89}, they presented a general method for converting a $k$-prover $\MIP$ into a $2$-prover $\MIP$, where in the resulting $2$-prover $\MIP$ the verifier sends one prover the queries $(q_1,\ldots,q_k)$ corresponding to {\em all} the $k$ provers in the underlying $k$-prover scheme, and expects to get back~$k$ answers $(a_1,\ldots,a_k)$; he sends the other prover a single query~$q_i$ corresponding to a random index $i\in[k]$, and gets back an answer~$a'_i$. The verifier accepts if and only if $a'_i=a_i$ and if the verifier in the $k$-prover $\MIP$ accepts the answers $(a_1,\ldots,a_k)$. 

In the non-signaling setting, we slightly modify this transformation by having the verifier of the $2$-prover $\MIP$ send the second prover a {\em subset} of queries $\{q_i\}_{i\in S}$ for a randomly chosen subset $S\subset [k]$ (as opposed to a single query $q_i$ corresponding to a single index $i\in [k]$), and accept if and only if the answers $(a_1,\ldots,a_k)$ of the first prover are accepted by the verifier of the $k$-prover $\MIP$ and if the answers of the second prover, denoted by $(a'_i)_{i\in S}$, satisfy $a'_i=a_i$ for every $i\in S$.
\begin{theorem}[Informal]\label{thm:intro:reduction}
There exist constants $c,d>0$ such that for every $k$-prover $\MIP$ $\Pi=(P_1,\ldots,P_k,V)$ with non-signaling soundness at most $2^{-ck^2}$, the 2-prover $\MIP$ obtained by performing the  prover reduction transformation (described above) on~$\Pi$ has non-signaling soundness at most $1 - 2^{-dk^2}$. 
\end{theorem}
We prove Theorem \ref{thm:intro:reduction} by using Theorem \ref{thm:intro:subNS-to-NS}.   We refer the reader to  Section~\ref{sec:intro:overview:proof:reduction} for the proof idea, and Section~\ref{sec:proof:reduction} for the precise theorem statement and proof. 

We next argue that  Theorem~\ref{thm:intro:reduction} implies Theorem~\ref{thm:intro:main}. Let $c,d$ be the constants from Theorem~\ref{thm:intro:reduction}.  We prove Theorem~\ref{thm:intro:main} with constants $c'=c$ and $d'=2d$.   To this end, fix any  $k$-prover $\MIP$ for a language~$L$ with no signaling soundness~$2^{-ck^2}$ and completeness $1-2^{-2dk^2}$.  Use Theorem~\ref{thm:intro:reduction} to convert this $\MIP$ into a $2$-prover $\MIP$ with non-signaling soundness $1-2^{-dk^2}$ (and completeness  $1-2^{-2dk^2}$). By \cite{Ito10}, the non-signaling value of any $2$-player game can be approximated up to an additive factor of $\epsilon$ in space $\poly(n,1/\epsilon)$.  Setting $\epsilon=2^{-2dk^2}$, there exists an algorithm $\mathcal{A}$ that runs in space $\poly(n,2^{k^2})$, such that on input an element  $x\in\{0,1\}^n\cap L$ it outputs a value $v\geq 1-2\cdot 2^{-2dk^2}$, and on input an element $x\in \{0,1\}^n\setminus L$ it outputs a value $v\leq 1-2^{-dk^2}+ 2^{-2dk^2}$.  This algorithm can be used to decide whether $x\in L$ (assuming without loss of generality that $d> \frac{2}{k^2}$), implying that $L\in\SPACE(\poly(n,2^{k^2}))$.

\paragraph{Approximating the sub-non-signaling value.}

We next present an alternative route for proving Theorem~\ref{thm:intro:main}, without going through the prover reduction method presented above.
Instead we prove the following theorem, which is of independent interest. 

\begin{theorem}[Informal]\label{thm:intro:subNS}
The sub-non-signaling value of any $k$-prover $\MIP$ with input length $n$, can be approximated up to an additive factor~$\epsilon$ by a 
$\poly(n,2^k,1/\epsilon,\cc)$-space algorithm, where $\cc$ is the communication complexity of the $\MIP$ on inputs of length~$n$. 
\end{theorem}  
In particular this theorem implies the following corollary.
\begin{corollary}[Informal]
$k$-prover $\subNS$ $\MIP$ is contained in $\SPACE\left(\poly(n,2^{k})\right)$.
\end{corollary}
See Section~\ref{sec:overview:proof:sub} for the proof idea, and see Section~\ref{sec:subNS} for the precise theorem statements and proofs. 
 We mention that a related (yet weaker) theorem was proven in~\cite{HY19}, where it was shown that given an $\MIP$, one can distinguish between the case that its classical value is~$1$ (i.e., there exists a local strategy that is accepted with probability~$1$) and the case that its sub-non-signaling value is at most $1-\delta$, in space $\poly(n,2^k,1/\delta)$.  This does not seem to be strong enough for us to use in order to obtain Theorem~\ref{thm:intro:main}.

We next argue that Theorem~\ref{thm:intro:subNS} and Theorem~\ref{thm:intro:subNS-to-NS} imply Theorem~\ref{thm:intro:main}.
To this end, let $c,d>0$ be the constants from Theorem~\ref{thm:intro:subNS-to-NS}.  We prove Theorem~\ref{thm:intro:main} with any constants $c',d'$ such that $c'>c$ and $d'=2d$.    Fix any $k$-prover $\MIP$ with soundness at most $2^{-c'k^2}<2^{-ck^2}$ and completeness at least $1-2^{2dk^2}$. By Theorem~\ref{thm:intro:subNS-to-NS}  for every $x\in \{0,1\}^n\setminus L$ the sub-non-signaling value of the $\MIP$ on input~$x$ must be less than $1-2^{-dk^2}$.  
By Theorem~\ref{thm:intro:subNS}, applied with $\epsilon=2^{-2dk^2}$, there exists an algorithm $\mathcal{A}$ that given any $x\in\{0,1\}^n$, runs in space $\poly(n,2^{k^2})$ and approximates the sub-non-signaling value of this $\MIP$ on input~$x$ up to an additive factor $2^{-2dk^2}$. Therefore for every $x\in \{0,1\}^n\setminus L$, the algorithm $\mathcal{A}$ outputs a value  $v\leq 1-2^{-dk^2}+2^{-2dk^2}$, and for every $x\in \{0,1\}^n\cap L$ the algorithm $\mathcal{A}(x)$ outputs an element  $v\geq 1-2\cdot 2^{-2dk^2}$.
This algorithm can be used to decide whether $x\in L$ (assuming without loss of generality that $d> \frac{2}{k^2}$), implying that $L\in\SPACE(\poly(n,2^{k^2}))$.




\section{Our Techniques}
In this section, we outline the high level overview of the proofs of Theorem~\ref{thm:intro:reduction} and Theorem~\ref{thm:intro:subNS} (the former uses Theorem~\ref{thm:intro:subNS-to-NS} as a building block). We defer the high level overview of the proof of Theorem~\ref{thm:intro:subNS-to-NS}, which contains the bulk of technical difficulty of this work, to Section~\ref{sec:intro:overview:proof:sub-to-ns}.\footnote{We defer this high-level overview since it is convenient to present it after the preliminaries section.  The reader can read Section~\ref{sec:prelim} and jump straight to Section~\ref{sec:intro:overview:proof:sub-to-ns} for the overview.}

\subsection{Overview of Theorem~\ref{thm:intro:reduction}
}\label{sec:intro:overview:proof:reduction}

The main ingredient in the proof of Theorem~\ref{thm:intro:reduction} is a claim showing that any non-signaling strategy for the $2$-prover $\MIP$ that succeeds in convincing the verifier to accept with probability $1-\epsilon$ can be converted into a {\em sub-non-signaling} strategy for the $k$-prover $\MIP$ that succeeds with probability $1-2^k\epsilon$.   This claim, together with Theorem~\ref{thm:intro:subNS-to-NS}, implies Theorem~\ref{thm:intro:reduction} in a relatively straightforward manner.  

We next provide the high-level overview of the proof of this claim.
Given a non-signaling strategy for the 2-prover $\MIP$ we construct a sub-non-signaling strategy for the $k$-party $\MIP$ as follows:  Given $q=(q_1,\ldots,q_k)$, run the non-signaling strategy for the 2-prover $\MIP$ $2^k$ times.  Namely, for every subset $S\subseteq[k]$, run the non-signaling prover for the $2$-prover $\MIP$, while giving the first prover all the queries~$q=(q_1,\ldots,q_k)$ and giving the second prover the subset $(q_i)_{i\in S}$.  If the verifier accepts the resulting answers in {\em all} the $2^k$ executions then output the answers given by the first prover in a random execution among these~$2^k$ executions.
Otherwise, if even one of these proofs is rejected then output~$\bot$.

One can easily argue that this strategy is accepted with probability $1-2^k\epsilon$ (by a straightforward application of the union bound).  Moreover, we argue that this strategy is sub-non-signaling. Intuitively, this follows from the fact that if all of the $2^k$ executions (of the $2$-prover $\MIP$) were accepting, then for every subset~$S$, the distribution of the answers  $(a_i)_{i\in S}$ is the same as the  distribution provided by the second prover in the $2$-prover $\MIP$ on input $(q_i)_{i\in S}$, which is non-signaling.  We refer the reader to Section~\ref{sec:proof:reduction} for the formal proof.

\subsection{Overview of the proof of Theorem~\ref{thm:intro:subNS}}\label{sec:overview:proof:sub}
The proof of this theorem follows the approach of~\cite{Ito10}, which proves that the non-signaling value of any 2-prover $\MIP$ can be approximated in $\PSPACE$.  Specifically, we define a linear program corresponding to the $k$-prover $\MIP$ such that the value of the linear program is equal to the sub-non-signaling value of the $\MIP$.  We then show that this linear program is of a specific form that allows it to be approximated in $\PSPACE$.  Specifically, we show that this linear program can be converted into a mixed packing and covering problem, and use the result of Young~\cite{Y01} which shows that such problems can be approximated via a space-efficient algorithm.  We refer the reader to Section~\ref{sec:subNS} for the precise theorems and proofs.

\section{Preliminaries}\label{sec:prelim}

\begin{definition}
A $k$-prover interactive proof for a language $L$ is said to have completeness $c$ if for every $x\in L$ the honest provers convince the verifier to accept $x\in L$ with probability at least~$c$.  
It is said to have soundness~$s$ is for every possibly malicious (non-interacting and local) provers, and for every $x\notin L$, the probability that these provers convince the verifier to accept $x\in L$ is at most~$s$.
\end{definition}
This work considers non-signaling cheating provers, as opposed to only local ones.  We also often think of the input~$x$ as fixed, and thus think of the proof system as a game, as opposed to a proof of membership in a language~$L$.
\subsection{Non-signaling Games}
\begin{definition}
A {\it $k$-prover, one-round game} is a tuple  $\mathcal{G}=(Q_1,...,Q_k,A_1,...,A_k,V,\pi)$, where $Q_1,...,Q_k$ are sets of  queries, $A_1,...,A_k$ are sets of answers, 
$$V: Q_1 \times Q_2 \times ... \times Q_k \times A_1 \times A_2 \times ... \times A_k \rightarrow \{0,1\}
$$ 
is a polynomial-time computable function, and $\pi$ is a polynomial-time sampleable probability distribution over $(Q_1,...,Q_k)$. 
\end{definition}
In the literature, the provers in a game are often referred to as players, and we use both interchangeably.

\paragraph{Notation.} We denote by $\mathcal{Q}\triangleq Q_1 \times Q_2 \times ... Q_k$ and  $\mathcal{A} \triangleq A_1 \times A_2 \times ... \times A_k$. We also denote by $\mathcal{Q}_{S} \triangleq Q_{s_1} \times Q_{s_2} \times ... \times Q_{s_{|S|}}$, where $S = \{s_1,s_2,...,s_{|S|}\}$, and similarly for $\mathcal{A}_{S}$. We denote by $[k]=\{1,\ldots,k\}$.  For every $q=(q_1,\ldots,q_k)\in\mathcal{Q}$, every $a=(a_1\ldots,a_k)\in\mathcal{A}$, and every $S\subseteq[k]$, we denote by $q_S=(q_i)_{i\in S}$ and $a_S=(a_i)_{i\in S}$.

\begin{definition}
A {\it strategy} for a game $\mathcal{G}=(\mathcal{Q},\mathcal{A},V,\pi)$ is a family of probability distributions $\{p_{q}\}_{q \in \mathcal{Q}}$ over $\mathcal{A}\cup\{\bot\}$. 
\end{definition}
For any $q\in\mathcal{Q}$ and $a\in\mathcal{A}$ we denote by $$
p_q(a)\triangleq\Pr[p_q=a],
$$
and for any subset $S \subseteq [k]$ we denote by $$p_q(a_S)\triangleq \sum_{a^*\in\mathcal{A}:a^*_S=a_S} p_q(a^*).$$
We use a similar notation for $p_q(\bot)$.

\begin{definition}\label{def:NS}
A strategy $\{p_{q}\}_{q\in\mathcal{Q}}$ for a $k$-player game $\mathcal{G}=(\mathcal{Q},\mathcal{A},V,\pi)$ is said to be {\it non-signaling} if  there exists a family of probability distributions $\{\Sim_{S,q_S}\}_{S\subseteq[k], q_S\in \mathcal{Q}_S}$, where each  $\Sim_{S,q_S}$ is a distribution over $\mathcal{A}_S$,  such that for every $q \in \mathcal{Q}$, every $S\subseteq [k]$, and every $a_{S} \in \mathcal{A}_{S}$, 
$$
p_{q}(a_{S}) = \Sim_{S,q_S}(a_S).
$$
\end{definition}
Namely, a strategy is non-signaling if the marginal distributions of the answers are the same regardless of the other queries. 
Note that if $\{p_q\}_{q\in\mathcal{Q}}$ is a non-signaling strategy then for every $q\in\mathcal{Q}$,
$$
\sum_{a_S}\Pr[p_q=a_S]=\sum_{a_S}\Pr[\Sim_{S,q_S}=a_S]=1,
$$
which implies that $\Pr[p_q=\bot]=0$. 

Two relaxations of the notion of non-signaling were considered in the literature:  the first is the notion of {\em sub-non-signaling}, by Lancien and Winter~\cite{LancienW15}, and the second is the notion of {\em honest-referee non-signaling} by Holmgren and Yang~\cite{HY19}. In both cases these relaxed notions were motivated by the goal of proving a parallel repetition theorem for non-signaling strategies.  We begin by defining the latter notion.

Loosely speaking, a strategy $\{p_{q}\}_{q\in\mathcal{Q}}$ for a $k$-player game $\mathcal{G}=(\mathcal{Q},\mathcal{A},V,\pi)$ is said to be {\it honest-referee non-signaling} if the non-signaling condition holds for every $q\in\mathcal{Q}$ such that $\Pr[\pi=q]>0$ (and is not required to hold for queries that are not in the support of~$\pi$).

\begin{definition}\label{def:hrNS} $\{p_{q}\}_{q\in\mathcal{Q}}$ is a honest-referee non-signaling strategy for $\mathcal{G}$ if  there exists a family of probability distributions  $\{\Sim_{S,q_S}\}_{S\subseteq[k], q_S\in \mathcal{Q}_S}$, where each $\Sim_{S,q_S}$ is a distribution over $\mathcal{A}_S$,  such that for every $q \in  \mathcal{Q}$ in the support of $\pi$, every $S\subseteq [k]$, and every $a_{S} \in \mathcal{A}_{S}$, 
$$
p_{q}(a_{S}) = \Sim_{S,q_S}(a_S).
$$
\end{definition}


\begin{definition}\label{def:subNS}
A strategy $\{p_{q}\}_{q\in\mathcal{Q}}$  for a $k$-player game $\mathcal{G}=(\mathcal{Q},\mathcal{A},V,\pi)$ is said to be   {\it sub-non-signaling} if  there exists a family of probability distributions $\{\Sim_{S,q_S}\}_{S\subseteq[k], q_S \in \mathcal{Q}_S}$, where each $\Sim_{S,q_S}$ is a distribution over $\mathcal{A}_S$, such that for every $q \in \mathcal{Q}$, every $S\subseteq [k]$, and every $a_{S} \in \mathcal{A}_{S} $, $$p_q(a_{S}) \leq \Sim_{S,q_S}(a_{S}).$$ 
\end{definition}

If $\{p_{q}\}_{q\in\mathcal{Q}}$ is a sub-non-signaling strategy then for every $q\in\mathcal{Q}$, if $$\sum_{a_S} p_q(a_{S}) <\sum_{a_S} \Sim_{S,q_S}(a_{S})=1,$$
then in the remaining probability $p_q$ outputs $\bot$.

\begin{definition}
Let $\NS(\mathcal{G})$ be the set of non-signaling strategies of a $k$-prover game $\mathcal{G}=(\mathcal{Q},\mathcal{A},V,\pi)$. The {\it non-signaling value} of $\mathcal{G}$ is 
$$
\mathcal{V}_{\NS}(\mathcal{G})=\max_{\{p_{q}\}_{q\in\mathcal{Q}} \in \NS(\mathcal{G})} \sum_{q\in\mathcal{Q}} \pi(q) \sum_{a\in\mathcal{A}} p_q(a) V(q,a).$$ Similarly, let $\hrNS(\mathcal{G})$ be the set of honest-referee non-signaling strategies of $\mathcal{G}$.  The  {\it honest-referee non-signaling value} of $\mathcal{G}$ is  $$
\mathcal{V}_{\hrNS}(\mathcal{G})=\max_{\{p_{q}\}_{q\in\mathcal{Q}} \in \hrNS(\mathcal{G})} \sum_{q\in\mathcal{Q}} \pi(q) \sum_{a\in\mathcal{A}} p_q(a) V(q,a).$$
Let $\subNS(\mathcal{G})$ be the set of  sub-non-signaling strategies of $\mathcal{G}$.  The  {\it sub-non-signaling value} of $\mathcal{G}$ is  $$
\mathcal{V}_{\subNS}(\mathcal{G})=\max_{\{p_{q}\}_{q\in\mathcal{Q}} \in \subNS(\mathcal{G})} \sum_{q\in\mathcal{Q}} \pi(q) \sum_{a\in\mathcal{A}} p_q(a) V(q,a).$$
\end{definition}

\begin{definition}
For any $\delta>0$ and any $k$-player game $\mathcal{G}=(\mathcal{Q},\mathcal{A},V,\pi)$, let $\subNS_\delta(\mathcal{G})$ be the set of all sub-non-signaling strategies $\{p_q\}_{q\in\mathcal{Q}}$ of the game $\mathcal{G}$ such that for every $q\in\mathcal{Q}$,
$$
\Pr[p_q=\bot]\leq \delta.
$$

\end{definition}

\iftrue{\subsection{Linear Programming}
\begin{definition}[\cite{MG07}]\label{def:dual}
Fix any linear program given by $\max \mathbf{c}^{\top}\mathbf{x}$ subject to $\mathbf{x}_S \geq 0$, $\mathbf{x}_T$ unrestricted, $A_{U}\mathbf{x} \leq \mathbf{b}_{U}$, and $A_{V}\mathbf{x} = \mathbf{b}_{V}$, where $S,T$ are disjoint and $S\cup T=[n]$, where $n=|x|$, and where $U,V$ are disjoint and $U\cup V=[m]$ where $m$ is the number of rows of~$A$, where $A$ is defined to be the matrix whose rows are the rows of $A_U$ and the rows of~$A_V$.    

The dual of this linear program is defined by $\min \mathbf{b}^{\top}\mathbf{y}$, where $|\mathbf{y}|=m$, subject to $\mathbf{y}_U \geq 0$, $\mathbf{y}_V$ unrestricted, $A^{\top}_{S}\mathbf{y} \geq \mathbf{c}_S$, $A^{\top}_{T}\mathbf{y} = \mathbf{c}_T$.
\end{definition}

\begin{theorem}[Strong duality \cite{MG07}]
If the value of a linear program is finite then it is equal to the value of its dual.
\end{theorem}

\begin{definition}[\cite{Y01}]
A mixed packing and covering problem is a pair of non-negative matrices $A,C$ and a pair of non-negative vectors $b,d$. A solution to a mixed packing and covering problem is a vector $x$ such that $x \geq 0$, $Ax \leq b$, and $Cx \geq d$. 
\end{definition}

\begin{theorem}[\cite{Y01}]
\label{packing}
Let $(A,b,C,d)$ be a mixed packing and covering problem. Then, there exists an algorithm running in space $\poly(\log(|(A,b,C,d)|),1/\epsilon)$ to determine whether there does not exist a solution to the mixed packing and covering problem or to output a solution to the mixed packing and covering problem $(A,b(1 + \epsilon),C,d)$. 
\end{theorem}
}\fi

\section{Non-Signaling Games with $k$ Players and $2^{-\Omega (k^2)}$ Soundness are in $\SPACE\left(\poly(n,2^{k^2})\right)$}

In what follows we state our main theorem.
\begin{theorem}\label{thm:NS}
There exists constants $c,d>0$ for which the following holds: Fix any language $L \notin \SPACE(\poly(n,2^k))$ and any $k$-prover one-round proof system $(P_1,\ldots,P_k,V)$ for $L$ with completeness $\geq 1-2^{-ck^2}$.  
For every $x$ consider the game $\mathcal{G}_x=(\mathcal{Q},\mathcal{A},V,\pi_x)$, where $\mathcal{Q}=\mathcal{Q}_1\times\ldots\times\mathcal{Q}_k$ and where $\mathcal{Q}_i$ is the set of possible queries sent by $V$ to prover $P_i$, $\mathcal{A}=\mathcal{A}_1\times,\ldots,\mathcal{A}_k$ and where $\mathcal{A}_i$ is the set of possible answers sent by $P_i$, and  $\pi_x$ is the distribution of queries sent by $V(x)$.

Then, there exists an infinite set $N\subseteq \mathbb{N}$, such that for every $n\in N$ there exists  $x\in\{0,1\}^n\setminus L$ such that $\mathcal{V}_{\NS}(\mathcal{G}_x)\geq 2^{-d\cdot k^2}$.

\end{theorem}

Our proof of Theorem~\ref{thm:NS} makes use of the following theorem which is the main technical contribution of this work. 
\begin{theorem}\label{thm:main}
There exist constants $c,d>0$, such that for any $k\in\mathbb{N}$ and any $k$-player game $\mathcal{G}$ the following holds:  If $\mathcal{V}_{\subNS}(\mathcal{G})\geq 1-2^{-ck^2}$  then $\mathcal{V}_{\NS}(\mathcal{G})\geq 2^{-d\cdot k^2}$. Moreover, for every $\delta\leq\frac{1}{k^{3k}}$, if $\mathcal{V}_{\subNS_\delta}(\mathcal{G})\geq 1-{\delta^2}$ then $\mathcal{V}_{\NS}(\mathcal{G})\geq \frac{1}{k^{3k}}$.   
\end{theorem}
We defer the proof of Theorem~\ref{thm:main} to Section~\ref{sec:proof:sub-to-ns}.
In what follows, we provide two alternative proofs for Theorem~\ref{thm:NS}, both which  use Theorem~\ref{thm:main} as a building block.  The first proof is given in Section~\ref{sec:proof:reduction} and the second proof is given in Section \ref{sec:subNS}. Both proofs only rely on the first part of Theorem~\ref{thm:main}.  The second part of Theorem~\ref{thm:main}, which converts a strategy in $\subNS_\delta$ into a non-signaling strategy, is not needed for our main result.  We add it as a contribution of independent interest, as it provides a tighter guarantee.   


\subsection{From Multi-Prover Non-Signaling Proofs to 2-Prover Non-Signaling Proofs}
\label{sec:proof:reduction} 

In the classical setting there is a well known reduction that converts any $k$-player game into a  $2$-player game.  Below we present a slight variant of it that will be useful in the non-signaling setting.  

Let $\mathcal{G}=(\mathcal{Q},\mathcal{A},V,\pi)$ be a $k$-player game. 
Consider the following $2$-player game, denoted by $\mathcal{T}(\mathcal{G})=(\mathcal{Q}^*,\mathcal{A}^*,V^*,\pi^*)$:

\begin{itemize}
    \item $\mathcal{Q}^*=(\mathcal{Q}^*_1,\mathcal{Q}^*_2)$, where $\mathcal{Q}^*_1=\mathcal{Q}$, $\mathcal{Q}^*_2=\{S,q_S\}_{S\subseteq[k],q_S\in\mathcal{Q}_S}$.
    \item $\mathcal{A}^*=(\mathcal{A}^*_1,\mathcal{A}^*_2)$, where $\mathcal{A}^*_1=\mathcal{A}$, $\mathcal{A}^*_2=\bigcup_{S\subseteq[k]} \mathcal{A}_S$.
    \item $\pi^*$ generates  $q\leftarrow \pi$ and generates a random subset $S\subseteq [k]$.  It outputs $(q,(S,q_S))$.
       \item $V^*((q, (S,q_S)), (a,a'_S)) $ accepts if and only if $V(q,a)$ accepts and  $a_i=a'_i$ for every $i\in S$. 
\end{itemize}

\begin{theorem}\label{thm:prover:reduction}
Let $c,d>0$ be the constants from Theorem~\ref{thm:main}.
Let $\mathcal{G}$ be a $k$-player game with non-signaling value less than~$2^{-dk^2}$.  Then the $2$-player game $\mathcal{T}(\mathcal{G})$ has non-signaling value at most $1-2^{-(c+1)k^2}$.
\end{theorem}

Before we prove Theorem~\ref{thm:prover:reduction}, we argue that it implies Theorem~\ref{thm:NS}. 
To see this, fix any $L$ and $(P_1,\ldots,P_k)$ as in the theorem statement.  Let $c,d$ be the constant from Theorem~\ref{thm:prover:reduction}.  We prove that   Theorem~\ref{thm:NS} holds with the constants $2c,d$. Suppose for contradiction that for every large enough $n\in \mathbb{N}$ and every $x\in \{0,1\}^n\setminus L$ it holds that $\mathcal{V}_\NS(\mathcal{G}_x)<2^{-dk^2}$, then by Theorem ~\ref{thm:prover:reduction}, $\mathcal{V}_\NS(\mathcal{T}(\mathcal{G}_x))\leq 1-2^{-(c+1)k^2}$, whereas for $x\in L$,  $\mathcal{V}_\NS(\mathcal{T}((\mathcal{G}_x))\geq 1-2^{-2ck^2}$.
By the work of Ito~\cite{Ito10}, this implies that $L\in\SPACE(\poly(n,2^k))$, contradicting our assumption.
\paragraph{Proof of Theorem~\ref{thm:prover:reduction}.}
Let $\mathcal{G}$ be a $k$-player game such that its non-signaling value is less than~$2^{-dk^2}$.
Suppose for the sake of contradiction that the non-signaling value of the 2-player game $\mathcal{T}(\mathcal{G})$ is  $1-\epsilon$, for $\epsilon< 2^{-(c+1)k^2}$.
Let $\{p_{q,(S,q_S)}\}$ be a non-signaling strategy that convinces the verifier $V^*$ in the game $\mathcal{T}({\mathcal{G}})$ to accept with probability $1-\epsilon$.

Consider the sub-non-signaling strategy $\{p_q\}$ for the $k$-player game $\mathcal{G}$, where $P_q$ samples answers as follows:  
\begin{enumerate}
    \item For every $S\subseteq [k]$, sample $(a,a'_S)\leftarrow p_{q,(S,q_S)}$.
    \item If there exists $S\subseteq [k]$ such that the above answers are rejecting (i.e., 
$V^*((q,(S,q_S)),(a,a'_S))=0$) then output $\bot$.
\item Otherwise, choose a random $S\subseteq [k]$ and output $a$ corresponding to this $S$.
\end{enumerate}   

\begin{claim}
$\{p_q\}$ is a sub-non-signaling strategy for the $k$-player game~$\mathcal{G}$.
\end{claim}

\begin{proof}
By definition, the fact that $\{p_{q,(S,q_S)}\}$ is a non-signaling distribution for the 2-player game $\mathcal{T}(\mathcal{G})$, implies that there is a family of distributions $\{\Sim_{q}\}\cup \{\Sim_{S,q_S}\}\cup\{\Sim_{q,(S,q_S)}\}$ such that for every $q\in\mathcal{Q}$, for every $S\subseteq [k]$ and every $a_S\in\mathcal{A}_S$,
$$
\Pr[p_{q,(S,q_S)}|_{(S,q_S)}=a_S]=\Pr[\Sim_{S,q_S}=a_S].
$$
We prove that $\{p_{q}\}$ is sub-non-signaling with respect to $\{\Sim_{S,q_S}\}$.  Namely, we prove that for every $q\in\mathcal{Q}$, every $S\subseteq[k]$, and every $a_S\in\mathcal{A}_S$,
\begin{equation}\label{eqn:MIP-2IP:subNS}
\Pr[p_{q}|_S=a_S]\leq \Pr[\Sim_{S,q_S}=a_S].
\end{equation}
We note that Equation~\eqref{eqn:MIP-2IP:subNS} would clearly hold if we chose $a$ corresponding to the specific set~$S$ in the equation.  However, recall that $p_q$ chooses $a$ corresponding to a {\em random} subset $S'\subseteq [k]$.

Thus, we define for every (fixed) $S\subseteq [k]$ a strategy $\{p^S_q\}$ which is identical to $\{p_q\}$, except that if it doesn't abort then it always outputs~$a$ corresponding to the fixed subset~$S$.
Therefore, to conclude the proof that $\{p_q\}$ is sub-non-signaling it suffices to prove that for every $q\in\mathcal{Q}$, every $a\in\mathcal{A}$, and every subsets $S,S'\subseteq[k]$, it holds that
$$
\Pr[p^S_q=a]=\Pr[p^{S'}_q=a], 
$$
which follows directly from the the fact that $\{p_{q,(S,q_S)}\}$ is non-signaling (together with the definition of $\{p^S_q\}$).  
\end{proof}

Note that the sub-non-signaling strategy $\{p_q\}$ is rejected with probability at most $2^k\cdot \epsilon$ (by the union bound). 
This in particular implies that the  sub-non-signaling value of $\mathcal{G}$ is at least  
$$1-2^k\cdot \epsilon\geq 1-2^k\cdot 2^{-(c+1)k^2}\geq 1-2^{-ck^2},$$ 
which by Theorem~\ref{thm:main} implies that the non-signaling value of $\mathcal{G}$ is at least $2^{-dk^2}$, contradicting our assumption.

\qed

\iftrue{
\subsection{Approximating the Sub-non-signaling Value of $k$-Player Game via a Space Efficient Algorithm}\label{sec:subNS}

\begin{theorem}\label{thm:LP'}
 There exists an algorithm $\mathcal{B}$ and a polynomial~$p$ such that for any  $k$-player game $\mathcal{G}=(\mathcal{Q},\mathcal{A},V,\pi)$, and any $\epsilon>0$, it holds that $\mathcal{B}(\mathcal{G},\epsilon)$ runs in space $p(\log(|\mathcal{Q},\mathcal{A}|),1/\epsilon, 2^k) $ and outputs a value $v$ such that $|v-\mathcal{V}_\subNS(\mathcal{G})|\leq \epsilon$.
\end{theorem}

\begin{corollary}\label{cor:LP} Fix any language $L$ and any $k$-prover one-round proof system $(P_1,\ldots,P_k,V)$ for~$L$.  For every $x$ consider the game $\mathcal{G}_x=(\mathcal{Q},\mathcal{A},V,\pi_x)$, where $\mathcal{Q}=\mathcal{Q}_1\times\ldots\times\mathcal{Q}_k$ and where $\mathcal{Q}_i$ is the set of possible queries sent by $V$ to prover $P_i$, $\mathcal{A}=\mathcal{A}_1\times,\ldots,\mathcal{A}_k$ where $\mathcal{A}_i$ is the set of possible answers sent by $P_i$, and  $\pi_x$ is the distribution of queries sent by $V(x)$.

Denote by~$c$ the completeness of this proof system. 
If there exists a constant $d\in\mathbb{N}$, such that for every large enough $n\in\mathbb{N}$, and every $x\in\{0,1\}^n\setminus L$, $\mathcal{V}_{\subNS}(\mathcal{G}_x)\leq c - \frac{1}{n^d}$, then $L \in \SPACE\left(\poly(n,2^k)\right)$.\footnote{This is assuming the communication complexity is $\poly(n)$.  In the general case, where the communication complexity is~$\cc$, we get that $L \in \SPACE\left(\poly(n,\cc,2^k)\right)$}
\end{corollary}

\paragraph{Proof of Corollary~\ref{cor:LP}.}

Fix any language $L$ and any $k$-prover one-round proof system $(P_1,\ldots,P_k,V)$ for~$L$ with completeness~$c$. For every $x\in\{0,1\}^*$, consider the corresponding game~$\mathcal{G}_x$ as defined in the corollary statement.  Suppose that there exists a constant $d\in\mathbb{N}$, such that for every large enough $n\in\mathbb{N}$, and every $x\in\{0,1\}^n\setminus L$, $\mathcal{V}_{\subNS}(\mathcal{G}_x)\leq c - \frac{1}{n^{d}}$.

Fix $\epsilon=\frac{1}{10\cdot n^d}$.  From Theorem \ref{thm:LP'} we know that there exists an algorithm $\mathcal{B}$,  that given any  $k$-prover game $\mathcal{G}=(\mathcal{Q},\mathcal{A},V,\pi)$, and any parameter $\epsilon$, approximates the value of~$\mathcal{G}$ up to an additive $\epsilon$ error.  Importantly $\mathcal{B}$ is an algorithm with space complexity  $\poly(\log(|(\mathcal{Q},\mathcal{A})|),1/\epsilon, 2^k) $. 

Given $x\in\{0,1\}^*$, we determine if $x\in L$ by running $\mathcal{B}(\mathcal{G}_x,1/\epsilon)$, and if the value is at least $c-\epsilon$ then we conclude that $x\in L$, and otherwise conclude that $x\notin L$. 

Note that $1/\epsilon$ is a polynomial in $n$ since $\epsilon = \frac{1}{10n^d}$. In addition, the size of $\mathcal{Q},\mathcal{A}$ is exponential in $n$, which implies that the space complexity of~$\mathcal{B}(\mathcal{G}_x,1/\epsilon)$ is $\poly(n,2^k)$, as desired.\footnote{More generally, if  $(P_1,\ldots,P_k,V)$ has communication complexity $\cc$ then $|(\mathcal{Q},\mathcal{A})|\leq 2^\cc$, in which case the space complexity of~$\mathcal{B}(\mathcal{G}_x,1/\epsilon)$ is $\poly(n,\cc,2^k)$, as desired. } 
Finally, we note that there may be 
a finite number of $n$'s for which we do not have the guarantee that $\mathcal{V}_{\subNS}(\mathcal{G}_x)\leq c - \frac{1}{n^d}$.  For these $n$'s, we can hard-wire the answers for whether $x \in L$. 
\qed

\medskip
We next prove Theorem~\ref{thm:LP'}.  We use the approach of \cite{Ito10} which proves that the non-signaling value of a two-player, one-round game can be approximated in PSPACE. The reason that \cite{Ito10} gives a result for non-signaling games is because using the linear program Ito shows that the non-signaling value and the sub-non-signaling value are equal for two-player games, which did not extend to games with more than two players.  

\paragraph{Proof of Theorem~\ref{thm:LP'}.}
 


Fix any game $\mathcal{G}=(\mathcal{A},\mathcal{Q},V,\pi)$.
The sub-non-signaling value of $\mathcal{G}$ is given by the following linear program (where the variables are ${p_q(a)}$ and $\Sim_{S,q_S}(a_{S})$, for every $q\in\mathcal {Q}$, $a\in\mathcal{A}$, and nonempty $S \subseteq[k]$)
\begin{equation}
\begin{array}{lll}
 \text{ Maximize } & \sum_{q\in\mathcal{Q}} \pi(q) \sum_{a\in\mathcal{A}} p_q(a) V(q,a) \\ \\

\text{ Subject to } & \sum_{a^*\in\mathcal{A}:a^*_S=a_S} p_{q}(a^*) ) \leq \Sim_{S,q_{S}}(a_{S})  &  \forall S \subseteq [k], \forall a_{S} \in \mathcal{A}_{S}, \forall q \in \mathcal{Q}, \\ \\

& \sum_{a_{S}\in\mathcal{A}_S} \Sim_{S,q_{S}}(a_{S} ) = 1  & \forall S \subseteq [k], \forall q_{S} \in \mathcal{Q}_{S} \\ \\

& p_{ q }(a) \geq 0  & \forall a \in \mathcal{A}, \forall q \in \mathcal{Q}
\end{array}
\end{equation}

In what follows, we replace $p_{q}(a)$ with $x_{q}(a) = \pi({q})p_{q}(a)$ to simplify the expression of the objective value. This gives us the linear program 
\begin{equation}
\begin{array}{lll}
 \text{ Maximize } & \sum_{q\in\mathcal{Q}} \sum_{a\in\mathcal{A}} x_{q}(a) V(q,a) \\ \\

\text{ Subject to } & \sum_{a^*\in\mathcal{A}:a^*_S=a_S} x_{q}(a^*) \leq \pi(q) \Sim_{S,q_{S}}(a_{S}  ) & \forall S \subseteq [k], \forall a_{S} \in \mathcal{A}_{S} \forall q, \in \mathcal{Q} \\ \\
& \sum_{a_{S}} \Sim_{S, q_{S}}(a_{S} ) = 1 & \forall S \subseteq [k], \forall q_{S} \in \mathcal{Q}_{S} \\ \\

& x_{q}(a) \geq 0 & \forall a \in \mathcal{A}, \forall q \in \mathcal{Q}
\end{array}
\end{equation}
Observe that the constraints in this linear program above imply that $\Sim_{S,q_S}(a_S)\geq 0$ for every $S\subseteq[k]$, every $q_S\in\mathcal{Q}_S$ and every $a_S\in\mathcal{A}_S$. Namely, these constraints can be added without changing the value of the linear program.   This implies (by  Definition~\ref{def:dual}), that the dual to this linear program can be written as

\begin{equation}
\begin{array}{lll}
\text{ Minimize } & \sum_{S \subseteq [k]} \sum_{q_{S} \in \mathcal{Q}_{S}} z_{S}(q_{S})  & \\ \\
\text{ Subject to } & \sum_{S \subseteq [k]} y_{S}(q,a_{S}) \geq V(q, a) & \forall q \in \mathcal{Q}, \forall a \in \mathcal{A} \\ \\
& z_{S}(q_{S}) \geq \sum_{q^*\in\mathcal{Q}: q^*_S=q_S}\pi(q^*) y_{S}(q^*,a_{S}) & \forall S \subseteq [k], \forall q_{S} \in \mathcal{Q}_{S} ,\forall a_{S} \in \mathcal{A}_{S} \\ \\
& y_{S}(q,a_{S}) \geq 0  & \forall S \subseteq [k] ,\forall q\in \mathcal{Q}, \forall a_{S} \in \mathcal{A}_{S} 
\end{array}
\end{equation}
Observe that the constraints in this linear program imply that $z_S(q_S) \geq 0$ for every $S\subseteq[k]$ and every $q_S\in\mathcal{Q}_S$, and thus these constraints can be added without changing the value. 

Next, transform this linear program into a linear program with non-negative coefficients. To do so, observe that the optimal solution to the above linear program satisfies that $y_{S}(q,a_S)\leq 1$ for every $S\subseteq[k]$, every $q\in\mathcal{Q}$ and every $a_S\in\mathcal{A}_S$.  This follows from the fact that $V(a,q)\leq 1$ for every $a\in\mathcal{A}$ and every $q\in\mathcal{Q}$.  Therefore, we can replace $y_{S}(q,a_S)$ by $\overline{y}_{S}(q,a_S) = 1 - y_{S}(q,a_S)$, without changing the value of the linear program. This gives us the linear program 
\begin{equation}
\begin{array}{lll}
\text{ Minimize } & \sum_{S \subseteq [k]} \sum_{q_{S} \in \mathcal{Q}_{S}} z_{S}(q_{S})  & \\ \\
\text{ Subject to } & \sum_{S \subseteq [k]} \overline{y}_{S}(q,a_{S}) \leq 2^k -1 - V( q, a )  & \forall q \in \mathcal{Q}, \forall a \in \mathcal{A} \\ \\
& z_{S}(q_{S}) + \sum_{q^*\in\mathcal{Q}:q^*_S=q_S} \pi(q^*) \overline{y}_{S}(q^*,a_{S}) \geq  \sum_{q^*\in\mathcal{Q}:q^*_S=q_S} \pi(q^*)  & \forall S \subseteq [k] ,\forall q_{S} \in \mathcal{Q}_{S} ,\forall a_{S} \in \mathcal{A}_{S} \\ \\
& \overline{y}_{S}(q,a_{S}) \leq 1 & \forall S \subseteq [k] ,\forall q \in \mathcal{Q} ,\forall a_{S} \in \mathcal{A}_{S} \\ \\
& \overline{y}_{S}(q,a_{S}) \geq 0 & \forall S \subseteq [k], \forall q \in \mathcal{Q}, \forall a_{S} \in \mathcal{A}_{S} \\ \\
& z_{S}(q_{S}) \geq 0 & \forall S \subseteq [k] , \forall q_{S} \in \mathcal{Q}_{S}
\end{array}
\end{equation}
Note that all of the coefficients of this linear program are non-negative. Also, because the parallel program takes parallel time polylogarithmic in the size of the linear program, it is not an issue that the linear program has size exponential in the input length.

Recall that our goal is to construct a  $\poly(\log(|(\mathcal{Q},\mathcal{A})|),1/\epsilon,2^k)$-space algorithm for computing~$v$ such that 
$$|v-\mathcal{V}_\subNS(\mathcal{G})|\leq \epsilon.$$ 
To this end, we add to our linear program a  constraint of the form
$$\sum_{S \subseteq [k]} \sum_{q_{S} \in \mathcal{Q}_{S}} z_{S}(q_{S}) \leq v'$$ 
(for some value~$v'$), and convert this 
(restricted) linear program into a mixed packing and covering program, with the guarantee that for $\delta = \frac{(\epsilon/2)}{2^k}$, a $(1+\delta)$-approximate solution to the mixed packing and covering program, implies a solution to the (restricted) linear program,  which is $\epsilon/2$-close an optimal solution. 
We can then use binary search to find an $\epsilon$-approximation to the original linear program.


To turn this restricted linear program into a mixed packing and covering problem, we use all of the constraints above and include the constraint $\sum_{S \subseteq [k]} \sum_{q_{S} \in \mathcal{Q}_{S}} z_{S}(q_{S}) \leq v'$.  

A $(1 + \delta)$-approximate solution to a mixed packing and covering problem is (by definition) a solution to the problem where all of the inequalities of the form $ a_ix_i \leq c$ are relaxed to $a_ix_i \leq c(1+\delta)$. In our case, it means that the above $\leq$ inequalities are replaced with $$\bar{y}_{S}(q,a_{S}) \leq 1 + \delta$$ 
and 
$$\sum_{S \subseteq [k]} \overline{y}_{S}(q,a_{S}) \leq (2^k - 1 - V(q,a)) (1 + \delta).$$ 
We next argue that a $(1+\delta)$-approximate solution to our mixed packing and covering problem implies a solution to our (restricted) linear program with value at most $v' + \epsilon$. 

To this end, suppose there these exists such a solution to the mixed packing and covering problem, and denote it by $$\left(\{y_{S}(q,a_S)\}_{S\in[k],q\in\mathcal{Q},a_S\in\mathcal{A}_S},\{z_{S}(q_S)\}_{S\in[k],q_S\in\mathcal{Q}_S}\right).$$ 
Consider the solution 
$$\left(\{y'_{S}(q,a_S)\}_{S\in[k],q\in\mathcal{Q},a_S\in\mathcal{A}_S},\{z'_{S}(q_S)\}_{S\in[k],q_S\in\mathcal{Q}_S}\right).$$ 
where 
$$y'_{S}(q,a_{S}) = \frac{1}{1 + \delta} y_{S}(q,a_{S})$$ 
and 
$$z'_{S}(q_S) = z_{S}(q_S) + \delta \sum_{q^*\in\mathcal{Q}:q^*_S=q_S}\pi(q^*).$$ 
It is easy to see that this solution satisfies the constraints of the (restricted) linear program, and thus is a solution to the linear program. 

The value of this solution is 
\begin{align*}
&\sum_{S \subseteq [k]} \sum_{q_{S} \in \mathcal{Q}_{S}} z'_{S}(q_{S}) = \\
&\sum_{S \subseteq [k]} \sum_{q_{S} \in \mathcal{Q}_{S}} \left(z_{S}(q_{S}) + \delta \sum_{q^*\in\mathcal{Q}^*:q^*_S=q_S} \pi(q^*)\right) = \\
& \sum_{S \subseteq [k]} \sum_{q_{S} \in \mathcal{Q}_{S}} z_{S}(q_{S})+\delta(2^k - 1) \leq \\
&v' + 2^{k}\delta < v' + \epsilon/2
\end{align*}

From Theorem \ref{packing} we can conclude that approximating the sub-non-signaling value of a game with a constant number of provers takes space polynomial in the log of the size of the linear program which is $\poly(|\mathcal{Q}|) \cdot \poly(|\mathcal{A}|) \cdot 2^k$ and in $1/\delta$, or $\poly(\log(|(\mathcal{Q},\mathcal{A})|),1/\epsilon, 2^k) $. 
\qed

\subsection{Proof of Theorem~\ref{thm:NS} via Corollary~\ref{cor:LP}}
In what follows we prove Theorem~\ref{thm:NS}.  In the proof we rely on Corollary~\ref{cor:LP} which implies that if $L\notin \SPACE(n,2^k)$ then there is an infinite set $N\subseteq \mathbb{N}$ such that for every $n\in N$ there is an element $x\in \{0,1\}^n\setminus L$ such that $\mathcal{V}_{\subNS}(\mathcal{G}_x)\geq c-\frac{1}{n}$.  Consider the infinite set $N_0\subseteq N$ such that for every $n\in N_0$ it holds that $n\geq 2^{5k^2}$.  We conclude that for every $n\in N_0$  there exists $x\notin \{0,1\}^n\setminus L$ such that $$\mathcal{V}_{\subNS}(\mathcal{G}_x)\geq  c-\frac{1}{n}\geq c- 2^{-5k^2}\geq  1-2^{-5k^2}-2^{-5k^2}\geq 1-2^{-4k^2}.$$
 Therefore, to prove Theorem~\ref{thm:NS} it suffices to prove the following theorem.
}\fi

\section{The Proof of Theorem~\ref{thm:main}}\label{sec:proof:sub-to-ns}

In this section we prove Theorem~\ref{thm:main}, which is our main technical theorem.  We start with the high-level overview of the proof.

\subsection{Overview of the proof of Theorem~\ref{thm:main}}\label{sec:intro:overview:proof:sub-to-ns}

In this overview we focus on proving the first part of Theorem~\ref{thm:main}, which is the part that contains the bulk of technical difficulty.  Namely, we need to show how to convert any sub-non-signaling strategy for a $k$-player game $\mathcal{G}=(\mathcal{Q},\mathcal{A},V,\pi)$ that convinces the verifier to accept with probability $1-2^{-ck^2}$ into a non-signaling strategy that convinces the verifier to accept  with probability~$2^{-dk^2}$ (for some constants $c,d>0$).  

To this end, we use the notion of {\em honest-referee non-signaling strategies}, defined by Holmgren and Yang~\cite{HY19} (see Definition~\ref{def:hrNS}). Loosely speaking, given any sub-non-signaling strategy $\{p_q\}_{q\in \mathcal{Q}}$ that succeeds in convincing $V$ to accept with probability $1-\epsilon$, we slightly modify the query distribution~$\pi$ into a new distribution~$\pi^*$ that is obtained by restricting $\pi$ to a subset of its domain~$\mathcal{Q}$, such that $\pi$ and $\pi^*$ are $\delta$-close, for an arbitrary parameter $\delta>0$ of our choice.  We construct an honest-referee non-signaling strategy with respect to~$\pi^*$ that convinces~$V$ to accept with probability at least $\frac{1}{k^{3k}}(1-k^{2k}\epsilon/\delta)$.  We then rely on a theorem from \cite{HY19} that shows how to convert an honest-referee non-signaling strategy that succeeds in convincing $V$ with probability $\eta$, into a non-signaling one that succeeds in convincing $V$ with probability $\geq 2^{-O(k^2)}\eta$ (see Theorem~\ref{thm:main:HR}). 

We note that if the sub-non-signaling strategy $\{p_q\}_{q\in \mathcal{Q}}$  is in $\subNS_\delta$ (for an appropriately small value of $\delta>0$) then our resulting honest-referee non-signaling strategy is in fact a non-signaling strategy, and hence we avoid the loss that is incurred by converting an honest-referee non-signaling strategy into a non-signaling one.

This is the reason we obtain a tighter bound in the second part of Theorem~\ref{thm:main}.

Fix any sub-non-signaling strategy  $\{p_q\}_{q\in \mathcal{Q}}$. By Definition~\ref{def:subNS}, there exists a set of distributions $\{\Sim_{S,q_S}\}_{S\subseteq[k],q_S\in\mathcal{Q}_S}$ such that for every $q\in\mathcal{Q}$, every $S\subseteq[k]$,  and every $a_S\in\mathcal{A}_S$,
$$
p_q(a_S)\leq \Sim_{S,q_S}(a_S).
$$
We show how to convert 
 the strategy $\{p_q\}_{q\in \mathcal{Q}}$ into an {\em honest-referee} non-signaling strategy via the following steps.
 
\begin{enumerate}
    \item 
{\bf Step 1.}\label{item:step1}    In this step we convert $\{\Sim_{S,q_S}\}$  into a family of distributions  $\{\Sim^{(1)}_{S,q_S}\}$, where each distribution $\Sim^{(1)}_{S,q_S}$ is over elements in $\mathcal{A}_S\cup\{\bot\}$,
such that 
for every $S,T\subseteq[k]$ for which $S\subseteq T$, and for every $q\in\mathcal{Q}$ and $a_S\in\mathcal{A}_S$,
\begin{equation}\label{eqn:overview:sim1}
\Pr[\Sim^{(1)}_{T,q_T}|_S=a_S]\leq \Pr[\Sim^{(1)}_{S,q_S}=a_S]
\end{equation}
and
\begin{equation}\label{eqn:overview-total}
\Pr_{q \leftarrow \pi, a\leftarrow \Sim^{(1)}_{[k],q}}[V(q,a)=1]\geq 1-\epsilon_1,
\end{equation}
where $\epsilon_1\triangleq k^{k}\cdot\epsilon$.
This is done via two sub-steps.
\begin{enumerate}
    \item We first reduce the probability of the ``outliers" of $p_q$.  Namely, if there exists a vector $q\in\mathcal{Q}$, a subset $S\subset [k]$, and answers $a_S\in\mathcal{A}_S$ such that  $\Pr[p_q|_S=a_S]$ is higher than the average probability over all $q^*$'s such that $q^*_S=q_S$, then we lower $\Pr[p_q|_S=a_S]$ towards the average, and in the remaining probability output~$\bot$.  Namely, we construct a family of distributions $\{\tilde{p}_q\}_{q\in \mathcal{Q}}$ such that for every $q\in\mathcal{Q}$, every $S\subseteq[k]$, and every $a_S\in\mathcal{A}_S$,
$$
\Pr[\tilde{p}_q|_S=a_S]\leq\E_{q^*\leftarrow{V}:q^*|_S=q_S}\Pr[p_q|_S=a_S].
$$
We note that ideally we would like to construct $\{\tilde{p}_q\}_{q\in \mathcal{Q}}$ that satisfies the above equation where the inequality is replaced with equality, since then  $\{\tilde{p}_q\}_{q\in \mathcal{Q}}$ would be non-signaling, and we would be done.  However, this is possible only if $\{p_q\}$ was non-signaling to begin with. Therefore, we start with the more humble goal of omitting the ``outliers".

We construct $\{\tilde{p}_q\}$ in a greedy manner, by starting with $\{p_q\}$ and then lowering the probabilities (in a greedy manner) so that the inequality above is satisfied.  Note that in the process we lower the total probability of $\tilde{p}_q$ (it outputs~$\bot$ in the remaining probability). However, we argue that the fact that $\{p_q\}$ is sub-non-signaling implies that the total probability is not reduced by too much. More specifically, we show that 
if 
$$\Pr_{q\leftarrow\pi, a\leftarrow p_q}[V(q,a)]= 1-\epsilon$$  then 
\begin{equation}\label{eqn:overview:delta}
\Pr_{q\leftarrow\pi, a\leftarrow \tilde{p}_q}[V(q,a)=1]\geq 1-2^{k}\epsilon.
\end{equation}

\item Define a family of distributions $\{\Sim'_{S,q_S}\}$ by $$
\Pr[\Sim'_{S,q_S}=a_S]\triangleq \max_{q^*:q^*_S=q_S} \Pr[\tilde{p}_{q^*}|_S=a_S],
$$
and in the remaining probability $\Sim'_{S,q_S}$ outputs~$\bot$.
At first it may seem that $\{\Sim'_{S,q_S}\}$ satisfies Equation~\eqref{eqn:overview:sim1}, since for a subset~$T$ that contains~$S$, we maximize over a smaller set of queries, and hence it may appear that the probability is smaller.  However, this is not quite true since  $\Pr[\Sim'_{T,q_T}|_S=a_S]$ is defined by summing over all $a_T$ that are consistent with $a_S$, the maximum  $$\max_{q^*:q^*_T=q_T} \Pr[\tilde{p}_{q^*}|_T=a_T],$$ which is larger than first maximizing and then summing.

Therefore, we ``correct" $\{\Sim'_{S,q_S}\}$ so that Equation~\eqref{eqn:overview:sim1} holds.  Specifically, we define 
 $\{\Sim^{(1)}_{S,q_S}\}$ in a greedy manner, by induction, as follows. 
For sets $S$ of size~$1$ and for every $q_S\in\mathcal{Q}_S$, define  $\Sim^{(1)}_{S,q_S}=\Sim'_{S,q_S}$.  Suppose we defined $\Sim^{(1)}_{S,q_S}$ for all sets $S$ of size less than~$i$, then for any set~$T$ of size $i$ and any $q_T\in\mathcal{Q}_T$, define  $\Sim^{(1)}_{T,q_T}$ in an iterative manner, as follows: 
Start by defining $\Sim^{(1)}_{T,q_T}=\Sim'_{T,q_T}$. 
If there exists $S\subsetneq T$ with $|S| = |T| - 1$ and a set $a_S\in\mathcal{A}_S$, such that
$$
\Pr[\Sim^{(1)}_{T,q_T}|_S=a_S]> \Pr[\Sim^{(1)}_{S,q_S}=a_S],
$$
then reduce the probability of  $\Sim^{(1)}_{T,q_T}$ so that 
$$
\Pr[\Sim^{(1)}_{T,q_T}|_S=a_S]= \Pr[\Sim^{(1)}_{S,q_S}=a_S],
$$
and in the remaining probability output~$\bot$.
This process ensures that indeed Equation~\eqref{eqn:overview:sim1} is satisfied.  However, it reduces the total probability of $\Sim^{(1)}_{[k],q}$, yet we argue that it does not reduce the probability by too much, and that indeed Equation~\eqref{eqn:overview-total} holds.

\end{enumerate}

\item 
{\bf Step 2.} We convert the family of distributions $\{\Sim^{(1)}_{S,q_S}\}$ into a family of honest-referee  non-signaling distributions.  This is done via the following two sub-steps.  

\begin{enumerate}
    
\item {\bf Step 2(a).} \label{item:step2}
We modify $\{\Sim^{(1)}_{S,q_S}\}$ to a new family of distributions $\{\Sim^{(2)}_{S,q_S}\}$ that still satisfies Equation~\eqref{eqn:overview:sim1}, yet in addition for every $S$ and $q_S$ the probability that $\Sim^{(2)}_{S,q_S}$ outputs~$\bot$ depends only on $|S|$, and is otherwise independent of~$S$ and $q_S$. 
We define $\Sim^{(2)}_{S,q_S}$ by lowering the probability mass of $\Sim^{(1)}_{S,q_S}$.  However, to ensure that we do not lower the probability mass by too much, we need to focus only on queries $q$ such that the probability that $\Sim^{(1)}_{[k],q}$ outputs~$\bot$ is low.  Specifically, in what follows, we focus only on queries~$q\in\mathcal{Q}$ such that 
$$
\Pr[\Sim^{(1)}_{[k],q}=\bot]\leq \epsilon_1/\delta,
$$
and we denote the set of all such queries by $\GOOD$. 
By Markov's inequality, together with Equation~\eqref{eqn:overview-total},
$$
\Pr[q\in\GOOD]\geq 1 - \delta.
$$
From now on we focus only on $q\in\GOOD$.
Namely, we consider the modified game where the queries are restricted to being in $\GOOD$.  Formally, we modify the game~$\mathcal{G}=(\mathcal{Q},\mathcal{A},V,\pi)$ by modifying the distribution~$\pi$ to the new distribution~$\pi^*= \pi|\GOOD$; i.e., $\pi^*$ samples~$q$ according to~$\pi$  subject to the restriction that $q\in\GOOD$.  From now on we focus on the game $\mathcal{G}^*$ where the distribution~$\pi$ is replaced with $\pi^*$.   We construct an honest-referee non-signaling strategy for this game. We mention that if the sub-non-signaling strategy $\{p_q\}$ is in $\subNS_\delta$ for $\delta< \sqrt{\epsilon}$ then $\GOOD=\mathcal{Q}$, and thus in this case the honest-referee non-signaling strategy is a non-signaling one, and thus we avoid the use of Theorem~\ref{thm:HY} and the loss associated with it. 


We first convert $\{\Sim^{(1)}_{S,q_S}\}$ into a new family of distributions $\{\Sim^{(2)}_{S,q_S}\}$ that still satisfies Equation~\eqref{eqn:overview:sim1},  but in addition it satisfies that for every $\ell\in[k]$, for $\alpha_\ell\triangleq (1-\epsilon_1/\delta)^\ell$,  for every $q\in\GOOD$ and every subset $S\subseteq [k]$ of size $\ell$, 
\begin{equation}\label{eqn:overview:**}
\sum_{a_S\in\mathcal{A}_S}\Pr[\Sim^{(2)}_{S,q_S}=a_S]=\alpha_\ell,
\end{equation}
and 
$$  \Pr_{q\leftarrow\pi^*, a\leftarrow \Sim^{(2)}_{[k],q}} [V(q,a)=1]
    \geq 1-(k+1)\epsilon_1/\delta\triangleq 1-\epsilon_2.
$$

This is done by simply normalizing each $\Sim^{(1)}_{S,q_S}$ accordingly.  We note that this normalization slightly reduces the success probability.  Nevertheless, Equation~\eqref{eqn:overview:**} is crucial, since it will allow us to use  $\{\Sim^{(2)}_{S,q_S}\}$ to construct a non-signaling strategy. 

\item {\bf Step 2(b).} \label{item:step3}
We next define an honest-referee non-signaling strategy for $\mathcal{G}^*$.  More specifically, we define a strategy for which the non-signaling condition holds for every query $q\in\GOOD$. 

For the sake of motivation, let's first try to define our honest-referee non-signaling strategy $\{p^*_q\}_{q\in\mathcal{Q}}$.  The idea is to define it in a greedy manner, as follows.  
We start by defining
$$
\Pr[p^*_q=a]=\Pr[\Sim^{(2)}_{[k],q}=a]
$$
We would like to argue that $$\Pr[p^*_q|_S=a_S]= \Pr[\Sim^{(2)}_{S,q_S}=a_S];$$ 
however, all we can guarantee is that \begin{equation}\label{eqn:overview:p*}
  \Pr[p^*_q|_S=a_S]\leq \Pr[\Sim^{(2)}_{S,q_S}=a_S].
\end{equation}

To remedy this, we modify the distribution $p^*_q$, as follows:
For every $i\in[k]$, let $\mathcal{A}^*_i=\mathcal{A}_i\cup\{*\}$, and let $\mathcal{A}^*=\mathcal{A}^*_1\times\ldots\times\mathcal{A}^*_k$. For each $q\in\mathcal{Q}$, we modify the distribution $p^*_q$ as follows:  We do not change its distribution over elements in $\mathcal{A}$, but we allow it to also output elements in $\mathcal{A}^*$ that are not in $\mathcal{A}$.  More specifically, we modify $p^*_q$ as follows:  
For any set $S\subseteq[k]$ and any $a_S\in\mathcal{A}_S$, we define
\begin{equation}\label{eqn:overviewi}
\Pr[p^*_q=(a_S,*)]=\Pr[\Sim^{(2)}_{S,q_S}=a_S]-\Pr[p^*_q|_S=a_S].
\end{equation}
Equation~\eqref{eqn:overview:p*} ensures that this probability is non-negative.  Moreover, Equation~\eqref{eqn:overviewi}
ensures that indeed
$$
\Pr[p^*_q|_S=a_S]=\Pr[\Sim^{(2)}_{S,q_S}=a_S],
$$
and hence only depends on $q_S$ as desired.

Unfortunately, this remedy does not work.  The reason is that only {\em initially} it is true that 
\begin{equation}\label{eqn:overview:non-neg}
\Pr[\Sim^{(2)}_{S,q_S}=a_S]-\Pr[p^*_q|_S=a_S]\geq 0.
\end{equation}
However,  as we modify the definition of~$p^*_q$, it's probability mass (i.e., $\sum_{a\in\mathcal{A}^*}{p^*_q(a)}$) grows, and can cause the left hand side in Equation~\eqref{eqn:overview:non-neg} to be negative!  

Instead, we first modify $\{\Sim^{(2)}_{S,q_S}\}$ into another family of distributions $\{\Sim^{(3)}_{S,q_S}\}$, which has the same desired properties as $\{\Sim^{(2)}_{S,q_S}\}$, but in addition satisfies \begin{equation}\label{eqn:overview:Sim3-prop}
\sum_{i=0}^{k-|S|} (-1)^{i} \sum_{T \supsetneq S, |T| = |S| + i}  \Pr[ \Sim^{(3)}_{T,q_T}|_S = a_S ]\geq 0. 
\end{equation}order to ensure that Equation~\eqref{eqn:overview:non-neg} remains non-negative.  
To ensure that the above equation is satisfies, we define for every $S\subseteq[k]$, every $q_S\in\mathcal{Q}_S$, and every $a_S\in\mathcal{A}_S$,
$$
Pr[\Sim^{(3)}_{S,q_S}=a_S]\triangleq \frac{1}{k^{2|S|}}\Pr[\Sim^{(2)}_{S,q_S}=a_S],
$$
and in the remaining probability it outputs~$\bot$.  We argue that indeed $\{\Sim^{(3)}_{S,q_S}\}$ satisfies Equation~\eqref{eqn:overview:Sim3-prop}.  Unfortunately, this step significantly reduces the acceptance probability, from one that approaches~$1$ as $k$ grows (with the right setting of parameters), to one that approaches $0$ as $k$ grows. Avoiding this loss is a great open problem.  

Equation~\eqref{eqn:overview:Sim3-prop} allows us to convert Equation~\eqref{eqn:overview:non-neg} to an equality, by setting for every non-empty subset $S\subseteq [k]$, every $q\in\mathcal{Q}$, and every $a_S\in\mathcal{A}_S$,
\begin{align*}
&\Pr[p^*_q=(a_S,(*)^{k-|S|)}]\triangleq \sum_{i=0}^{k-|S|}(-1)^{i}\sum_{T\supseteq S, |T|=|S|+i}\Pr[\Sim^{(3)}_{T,q_T}|_S=a_S].
\end{align*}
This is exactly the extension needed in order to convert the inequality in Equation~\eqref{eqn:overview:non-neg} to an equality, and by definition of $\Sim^{(3)}$ (and in particular, by Equation~\eqref{eqn:overview:Sim3-prop}), it is always the case that 
$$
\Pr[p^*_q=(a_S,(*)^{k-|S|})]\geq 0.
$$

Finally, we note that by defining $\{p^*_q\}_{q\in\mathcal{Q}}$ as above, the total probability of $p^*_q$ may not be exactly~$1$.  It may be smaller than~$1$ or greater than~$1$.  However, its total probability is fixed and does not depend on~$q$.  Therefore, we can safely normalize it to be exactly~$1$ without damaging the honest-referee non-signaling guarantee.

\end{enumerate}

\end{enumerate}


\paragraph{Our Parameters.}  Recall that we convert a sub-non-signaling strategy with value at least $1-2^{-dk^2}$ into a non-signaling strategy with value at least $2^{-ck^2}$ (for some constants $c,d>0$).  We don't have any reason to believe that this loss is inherent.  In particular, we would like to convert any sub-non-signaling strategy with value at least $1-2^{-dk}$ into a non-signaling strategy with value at least $2^{-ck}$.  This would imply that $O(\log n)$-prover non-signaling $\MIP$ is in $\PSPACE$. Our loss stems mainly from the step where we go from honest-referee non-signaling to non-signaling, via a transformation from~\cite{HY19}.  There is another $k^{k}$ loss in Step~\ref{item:step3}, however this loss is small compared to the other one.  We do not know if these losses are inherent, and leave it as an open problem to explore. 
\subsection{Formal Proof of Theorem~\ref{thm:main}}\label{sec:proof:main}
We prove Theorem~\ref{thm:main} by proving the following theorem.
\begin{theorem}\label{thm:main:HR}
For any $k\in\mathbb{N}$, any $k$-player game $=(\mathcal{Q},\mathcal{A},V,\pi)$, any $\epsilon>0$ and any $\delta>0$:  If $\mathcal{V}_\subNS(\mathcal{G})\geq 1-\epsilon$ then there exists a game $\mathcal{G}^*=(\mathcal{Q},\mathcal{A}^*,V^*,\pi^*)$ such that the following holds:
\begin{enumerate}
    \item For every $i\in[k]$, $\mathcal{A}^*_i=\mathcal{A}_i\cup \{*\}$, where $\mathcal{A}=\mathcal{A}_1\times\ldots,\times\mathcal{A}_k$ and $\mathcal{A}^*=\mathcal{A}_1^*\times\ldots,\times\mathcal{A}_k^*$.  
    \item $V^*|_{\mathcal{Q}\times\mathcal{A}}\equiv V$ and $V^*|_{\mathcal{Q}\times(\mathcal{A}^*\setminus\mathcal{A})}\equiv 0$
    \item $\pi$ and $\pi^*$ are $\delta$-close.
    \item $\mathcal{V}_\hrNS(\mathcal{G}^*)\geq \frac{1}{k^{3k}}\left(1-k^{2k}\epsilon/\delta\right)$.
\end{enumerate} 
Moreover, if $\mathcal{V}_{\subNS_\delta}(\mathcal{G})\geq 1-\epsilon$ and $\delta\leq \sqrt{\epsilon}$ then $\pi^*=\pi$ and $\mathcal{V}_\NS(\mathcal{G}^*)\geq \frac{1}{k^{3k}}\left(1-k^{2k}\epsilon/\delta\right)$.
\end{theorem}
We use this theorem together with the theorem from~\cite{HY19} listed below to prove Theorem~\ref{thm:main}.  
\begin{theorem}\cite{HY19}\label{thm:HY}
For every $k\in\mathbb{N}$ there exists a fixed value $\alpha_k\geq 2^{-O(k^2)}$ such that for any $k$-player game $\mathcal{G}$,  $\mathcal{V}_{\NS}(\mathcal{G})\geq \alpha_k\cdot \mathcal{V}_{\hrNS}(\mathcal{G})$.
\end{theorem}

\paragraph{Proof of Theorem~\ref{thm:main}.} 
Let $c>0$ be a constant such that $2^{-ck^2}= \frac{\alpha_k}{k^{8k}}$, and let $d>0$ be a constant such that $2^{-dk^2}= \frac{\alpha_k}{k^{5k}}$, where $\alpha_k$ is the fixed value given in Theorem~\ref{thm:HY}.
Let $\mathcal{G}= (\mathcal{Q},\mathcal{A},V,\pi)$ be any $k$-player game such that $\mathcal{V}_{\subNS}(\mathcal{G})\geq 1 - \epsilon$.  Let $\delta=\epsilon\cdot k^{3k}$.  By Theorem~\ref{thm:main:HR} there exists a game $\mathcal{G}^*=(\mathcal{Q},\mathcal{A}^*,V^*,\pi^*)$ that satisfies the conditions of Theorem~\ref{thm:main:HR}, and in particular 
$$
\mathcal{V}_\hrNS(\mathcal{G}^*)\geq 
\frac{1}{k^{3k}}\left(1-k^{2k}\epsilon/\delta\right)= \frac{1}{k^{3k}}\left(1-k^{-k}\right)\geq \frac{1}{k^{4k}}.
$$
This, together with  Theorem~\ref{thm:HY}, 
implies that
$$\mathcal{V}_\NS(\mathcal{G}^*)\geq  \frac{\alpha_k}{k^{4k}}$$
Set $\epsilon=2^{-ck^2}=\frac{\alpha_k}{k^{8k}}$, which implies that  $\delta=\frac{\alpha_k}{k^{5k}}$.
We next argue that with this setting of parameters,  
$$\mathcal{V}_\NS(\mathcal{G})\geq  \frac{\alpha_k}{k^{5k}}=2^{-dk^2}, $$
as desired.
This follows from the following two steps:  First  convert the non-signaling strategy for $\mathcal{G}^*$ into a non-signaling strategy where all the answers are in $\mathcal{A}$ (as opposed to $\mathcal{A}^*$).  This is done as follows:  Arbitrarily choose a fixed tuple $(a_1,\ldots,a_k)\in \mathcal{A}_1\times\ldots,\times\mathcal{A}_k$.  If the answer in the $i$'th coordinate is $*\in \mathcal{A}_i^*\setminus \mathcal{A}_i$ then replace it with the fixed answer $a_i\in\mathcal{A}_i$.  Note that this new strategy remains non-signaling.  Moreover, the fact that $V^*$ always rejects the answers that are not in $\mathcal{A}$, implies that the value of this non-signaling strategy in $\mathcal{G}^*$ does not decrease.  
Finally, the fact that $\pi^*$ and $\pi$ are $\delta$-close implies that indeed
$$\mathcal{V}_\NS(\mathcal{G})\geq \mathcal{V}_\NS(\mathcal{G}^*)-\delta\geq  \frac{\alpha}{k^{4k}}- \frac{\alpha_k}{k^{5k}}\geq \frac{\alpha_k}{k^{5k}}.
$$
We next prove the second part of Theorem~\ref{thm:main}.  To this end, fix any $\delta\leq\frac{1}{k^{3k}}$,  and set $\epsilon =\frac{1}{k^{6k}}\geq \delta^2$. 
Theorem~\ref{thm:main:HR}  implies that for this setting of parameters
$$
\mathcal{V}_\NS(\mathcal{G}^*)\geq 
\frac{1}{k^{3k}}\left(1-k^{2k}\epsilon/\delta\right)= \frac{1}{k^{3k}}\left(1-k^{-k}\right)\geq \frac{1}{k^{4k}}.
$$
where as above this implies that 
$$
\mathcal{V}_\NS(\mathcal{G})\geq \mathcal{V}_\NS(\mathcal{G}^*)-\delta\geq  \frac{1}{k^{4k}}-\frac{1}{k^{3k}}\geq \frac{1}{k^{3k}},
$$
as desired.
\qed

\paragraph{Proof of Theorem~\ref{thm:main:HR}.}
Let $\mathcal{G}= (\mathcal{Q},\mathcal{A},V,\pi)$ be a $k$-player game such that $\mathcal{V}_{\subNS}(\mathcal{G})\geq 1 - \epsilon$. 
Let $\{p_{q}\}_{q\in\mathcal{Q}}$ be a sub-non-signaling strategy, such that $\mathcal{G}$ has sub-non-signaling value $1-\epsilon$ with respect to $\{p_{q}\}_{q\in\mathcal{Q}}$. 
In what follows, we denote by 
$$
p_{q_S}(\bot)\triangleq\E_{q^*\leftarrow\pi|(q^*_S=q_S)}[p_{q^*}(\bot)],
$$
and we denote by 
\begin{equation}\label{eqn:nu}
\nu(q)\triangleq\sum_{S\subseteq[k]} p_{q_S}(\bot).
\end{equation}
Note that
\begin{equation}\label{eqn:nu-avg}
\E_{q\leftarrow\pi}[\nu(q)]\leq \sum_{S\subseteq[k]} \E_{q\leftarrow\pi}[p_{q_S}(\bot)]\leq 2^k\cdot\epsilon.
\end{equation}
Our proof proceeds in two steps, each which consists of two sub-steps.

  \paragraph{Step 1.}  Construct a family of distributions $\{\Sim^{(1)}_{S,q_S}\}$ over $\mathcal{A}_S\cup\{\bot\}$,  such that for every $S\subseteq T\subseteq [k]$, and every $q_T\in\mathcal{Q}_T$ and $a_S\in\mathcal{A}_S$, it holds that 
  \begin{equation}\label{eqn:*}
  \Pr[\Sim^{(1)}_{T,q_T}|_S = a_S]\leq \Pr[\Sim^{(1)}_{S,q_S} = a_S],
  \end{equation}
  and for every $q\in\mathcal{Q}$ and $S\subseteq[k]$
 \begin{equation}\label{eqn:*2}
  \Pr[\Sim^{(1)}_{S,q_S}=\bot]\leq (|S|+2)!\cdot\nu(q), \end{equation} and 
  \begin{equation}\label{eqn:sum-Sim*}
  \Pr_{q\leftarrow \pi,a\leftarrow\Sim^{(1)}_{[k],q}}[V(q,a)=1]\geq 1-k^{\log k}\cdot\epsilon.
  \end{equation}
   We do this in two steps.  

\paragraph{Step 1(a).}  We define a sub-non-signaling strategy $\{\tilde{p}_q\}$ for the game $\mathcal{G}$, such that for every $q\in\mathcal{Q}$, \begin{equation}\label{eqn:tildep1}
\tilde{p}_q(a)\leq p_q(a) \mbox{  }\mbox{ } \forall a\in\mathcal{A} 
\end{equation}
and  
\begin{equation}\label{eqn:tildep2}
\tilde{p}_q(\bot)\leq 
\nu(q),
\end{equation}
and in addition
for every $S \subseteq [k]$ and every  $a_{S}\in\mathcal{A}_S$
\begin{equation}\label{eqn:less-than-avg}
\Pr[\tilde{p}_{q}|_S = a_S] \leq \E_{q^* \leftarrow \pi|(q^*_S = q_S)} \Pr[ p_{q^*}|_{S} = a_S ] 
\end{equation} 
Note that Equations~\eqref{eqn:tildep1} and~\eqref{eqn:tildep2} imply that for every $q\in\mathcal{Q}$
\begin{equation}\label{eqn:tilde-succ}
\Pr_{a\leftarrow\tilde{p}_q }[V(q,a)=1]\geq \Pr_{a\leftarrow p_q }[V(q,a)=1]-\nu(q).
\end{equation}
We define $\tilde{p}_q$ in a greedy manner, so that
Equation~\eqref{eqn:less-than-avg} holds,
while keeping the invariant that Equation~\eqref{eqn:tildep1} holds.   This is done as follows:  Fix any $q\in\mathcal{Q}$. Start with $\tilde{p}_q=p_q$.
For every $S\subseteq [k]$ and every $a_S$, if 
$$
\Pr[\tilde{p}_{q}|_S = a_S] > \E_{q^* \leftarrow \pi|(q^*_S = q_S)} \Pr[ p_{q^*}|_{S} = a_S ] 
$$
then (arbitrarily) reduce $\tilde{p}_q(a^*)$ for every $a^*\in\mathcal{A}$ such that $a^*_S=a_S$ so that $$
\Pr[\tilde{p}_{q}|_S = a_S] = \E_{q^* \leftarrow \pi|(q^*_S = q_S)} \Pr[ p_{q^*}|_{S} = a_S ] ,
$$
and in the remaining probability output~$\bot$.
For each $S$ and $a_S$, this step reduces the probability that~$V$ accepts by at most
$$
\delta_{S,a_S}(q)\triangleq \max 0, \Pr[{p}_{q}|_S = a_S] - \E_{q^* \leftarrow \pi^*| (q^*_S = q_S)} \Pr[ p_{q^*}|_{S} = a_S ] .
$$
This follows from the invariant that for every $a$ it holds that $\tilde{p}_q(a)\leq p_q(a)$.
Since we do this for every $S\subseteq [k]$ and every $a_S$, in total the probability of $\bot$ is increased by at most
$$\delta(q) = \sum_{S,a_S} \delta_{S,a_S}(q).$$
Note that 
Equations~\eqref{eqn:tildep1} and~\eqref{eqn:less-than-avg} hold by definition of $\{\tilde{p}_q\}$. To prove Equation~\eqref{eqn:tildep2}, it suffices to prove the following claim.
\begin{claim}\label{claim:delta}
For every $q\in\mathcal{Q}$, it holds that $\delta(q)+p_q(\bot) \leq \nu(q)$. 
\end{claim}

\begin{proof}
Since $\{p_q\}$ is a sub-non-signaling strategy, there exists a family of distributions $\{\Sim_{S,q_S}\}$ such that for every $S\subseteq[k]$ and every $q_S\in\mathcal{Q}_S$ and $a_S\in\mathcal{A}_S$, 
$$\max_{q^*\mbox{ }s.t.\mbox{ } q^*_S = q_S} \Pr[p_{q^*}|_S = a_S ] \leq \Pr[ \Sim_{S,q_S} = a_S ]. $$ Therefore,
$$\Pr[ \Sim_{S,q_S} = a_S ] \geq \E_{q^* \leftarrow \pi|(q^*_S = q_S)} \Pr[ p_{q^*}|_{S} = a_S ] + \delta_{S,a_S}(q),
$$
which implies that 
$$
1 = \sum_{a_S} \Pr[ \Sim_{S,q_S} = a_S ] \geq 
\sum_{a_S} \E_{q^* \leftarrow \pi|(q^*_S = q_S)} \Pr[ p_{q^*}|_{S} = a_S ] + \sum_{a_S} \delta_{S,a_S}(q) \geq 1-\E_{q^* \leftarrow \pi|(q^*_S = q_S)} [p_{q^*}(\bot)] + \sum_{a_S} \delta_{S,a_S}(q).$$
We thus conclude that for every $q\in\mathcal{Q}$ and for every $S\subsetneq [k]$, $\sum_{a_S} \delta_{S,a_S}(q) \leq p_{q_S}(\bot)$,
and $\sum_{a}\delta_{[k],a}(q)=0$.  This in turn implies that $\delta(q)+p_q(\bot)=\sum_{S\subsetneq[k],a_S} \delta_{S,a_S}(q)+p_q(\bot) \leq \nu(q)$, as desired. 
\end{proof}
Thus, the strategy $\{\tilde{p}_q\}$ satisfies Equations~\eqref{eqn:tildep1},~\eqref{eqn:tildep2} and~\eqref{eqn:less-than-avg} (and as a result it also satisfies Equation~\eqref{eqn:tilde-succ}).

\paragraph{Step 1(b).}  We next define the family of distributions $\{\Sim^{(1)}_{S,q_S}\}$  over $\mathcal{A}_S\cup\{\bot\}$ that satisfies Equations~\eqref{eqn:*},~\eqref{eqn:*2} and \eqref{eqn:sum-Sim*}.



We start by defining $\{\Sim'_{S,q_S}\}$ by
$$
\Pr[ \Sim'_{S,q_S} = a_S] \triangleq \max_{q^*\in\mathcal{Q}| (q^*_S = q_S)}  \Pr[ \tilde{p}_{q^*}|_S = a_S].
$$
Note that $\sum_{a_S} \Pr[ \Sim'_{S,q_S} = a_S]\leq 1$ since by  Equation~\eqref{eqn:less-than-avg},
$$
\Pr[ \Sim'_{S,q_S} = a_S] = \max_{q^*\in\mathcal{Q}| (q^*_S = q_S)}  \Pr[ \tilde{p}_{q^*}|_S = a_S] \leq  \E_{q^*\leftarrow \pi|(q^*_S=q_S)}\Pr[p_{q^*}|_S=a_S],
$$
which together with the linearity of expectation, implies that indeed 
$$
\sum_{a_S\in\mathcal{A}_S}\Pr[ \Sim'_{S,q_S} = a_S] \leq \sum_{a_S\in\mathcal{A}_S}\E_{q^*\leftarrow \pi|(q^*_S=q_S)}\Pr[p_{q^*}|_S=a_S= \E_{q^*\leftarrow \pi|(q^*_S=q_S)}\sum_{a_S\in\mathcal{A}_S}\Pr[p_{q^*}|_S=a_S]\leq 1.
$$
Moreover, Equation~\eqref{eqn:tildep2}, together with the definition of $\{\Sim'_{S,q_S}\}$, implies that for every $q\in\mathcal{Q}$ and every $S\subseteq[k]$,
\begin{equation}\label{eqn:Sim'-bot}
\Pr[\Sim'_{S,q_S}=\bot]\leq \nu(q),
\end{equation}
and Equation~\eqref{eqn:tilde-succ} implies that  \begin{equation}\label{eqn:sum-Sim*-int}
\Pr_{a\leftarrow \Sim'_{[k],q}} [V(q,a)=1]\geq \Pr_{a\leftarrow p_q} [V(q,a)=1]-\nu(q).
\end{equation}
We next define $\{\Sim^{(1)}_{S,q_S}\}$ by modifying $\{\Sim'_{S,q_S}\}$ in a greedy manner, to ensure that Equation~\eqref{eqn:*} is satisfied.  This is done
by induction starting with sets of size~$1$.  For every set~$T$ of size~1, and for every~$q_T$, define 
$$
\Sim^{(1)}_{T,q_T}\triangleq\Sim'_{T,q_T}.
$$
Suppose we defined $\Sim^{(1)}_{S,q_S}$ for all  sets~$S$ of size less than~$i$. We next define $\Sim^{(1)}_{T,q_T}$ for sets $T$ of size~$i$.
To this end, fix any 
$T$ of size~$i$ and fix any $q_T$. Start by setting
$$
\Sim^{(1)}_{T,q_T}=\Sim'_{T,q_T}.
$$
For every $S\subset T$ of size $i-1$ and for every $a_S$, if
$$
\Pr[\Sim^{(1)}_{T,q_T}|_S=a_S]> \Pr[\Sim^{(1)}_{S,q_S}=a_S]
$$
then (arbitrarily) reduce the probability that  $\Pr[\Sim^{(1)}_{T,q_T}|_S=a_S]$ by exactly  $$\xi_{T,q_T}(S,a_S)\triangleq\Pr[\Sim^{(1)}_{T,q_T}|_S=a_S]-\Pr[\Sim^{(1)}_{S,q_S}=a_S],$$ so that 
\begin{equation}\label{eqn:sim1-ind}
\Pr[\Sim^{(1)}_{T,q_T}|_S=a_S]= \Pr[\Sim^{(1)}_{S,q_S}=a_S].
\end{equation}
In the remaining probability output~$\bot$.
We next argue that this ensures that Equation~\eqref{eqn:*} holds.  We prove this by induction on the size of $T$.  Clearly Equation~\eqref{eqn:*} holds for sets $T$ of size~$2$.  Suppose Equation~\eqref{eqn:*} holds for sets $T$ of size $i-1$ and we prove that it holds for sets $T$ of size $i$.  To this end, fix any set $T$ of size $i$ and any $S\subset T$. Let $S'$ be an arbitrary set of size $i-1$ such that $S\subseteq S'\subset T$. By Equation~\eqref{eqn:sim1-ind}, for every $a\in\mathcal{A}$
$$
\Pr[\Sim^{(1)}_{T,q_T}|_{S'}=a_{S'}]\leq \Pr[\Sim^{(1)}_{S',q_{S'}}=a_{S'}] 
$$
and by our induction hypothesis,
$$
 \Pr[\Sim^{(1)}_{S',q_{S'}}|_S=a_{S}]\leq \Pr[\Sim^{(1)}_{S,q_{S}}=a_{S}].
 $$
These two equations imply that 
\begin{align*}
&\Pr[\Sim^{(1)}_{T,q_{T}}|_S=a_{S}]=\\
&\sum_{a_{S'}:{a_{S'}|_S=a_S}}\Pr[\Sim^{(1)}_{T,q_{T}}|_{S'}=a_{S'}]\leq\\
&\sum_{a_{S'}:{a_{S'}|_S=a_S}}\Pr[\Sim^{(1)}_{S',q_{S'}}=a_{S'}]=\\
&\Pr[\Sim^{(1)}_{S',q_{S'}}|_S=a_{S}]\leq\\
&\Pr[\Sim^{(1)}_{S,q_{S}}=a_{S}],
\end{align*}
as desired.

We next argue that despite this reduction in probability,  Equations~\eqref{eqn:*2} and~\eqref{eqn:sum-Sim*} hold.  
To this end, note that for every $S,T\subseteq[k]$ such that $S\subset T$ and $|S|=|T|-1$, and every $q\in\mathcal{Q}$ and $a\in\mathcal{A}$,
$$
\xi_{T,q_T}(S,a_S)\leq\max\left\{0,\Pr[\Sim'_{T,q_T}|_S=a_S]- \Pr[\Sim^{(1)}_{S,q_S}=a_S]\right\}.
$$
Define 
$$
\xi_{T,q_T}(S)\triangleq \sum_{a_S}\xi_{T,q_T}(S,a_S)\mbox{ }\mbox{ and }\mbox{ }\xi_{T,q_T}\triangleq \sum_{S\subsetneq T: |S|=|T|-1}\xi_{T,q_T}(S).
$$

\begin{claim}\label{claim:xi}
For every $q\in\mathcal{Q}$ and every $T\subseteq [k]$
$$
\xi_{T,q_T}\leq (|T|+1)!\cdot \nu(q).
$$
\end{claim}
Note that Claim~\ref{claim:xi}, together with the definition of $\{\Sim^{(1)}_{[k],q}\}$ and with Equation~\eqref{eqn:Sim'-bot}, implies that 
Equation~\eqref{eqn:*2} holds.
Similarly,  Claim~\ref{claim:xi}, together with Equation~\eqref{eqn:sum-Sim*-int}, implies that Equation~\eqref{eqn:sum-Sim*} holds, since
\begin{align*}
&\Pr_{q\leftarrow\pi, a\leftarrow \Sim^{(1)}_{[k],q}}[V(q,a)=1]\geq\\ &\Pr_{q\leftarrow\pi,a\leftarrow p_q}[V(q,a)=1]-\E_{q\leftarrow \pi} [\nu(q)+(k+1)!\cdot\nu(q)]\geq\\ &1-\epsilon-((k+1)!+1)\cdot \E_{q\leftarrow \pi}[\nu(q)]\geq 1-2^{k\log k} \cdot\epsilon,
\end{align*}
where the latter inequality follows from Equation~\eqref{eqn:nu-avg}.

\paragraph{Proof of Claim~\ref{claim:xi}.}
Fix any $q\in\mathcal{Q}$ and any $T\subseteq[k]$.  Note that by definition for every $S\subset T$ of size $|T|-1$ and for every $a_S\in\mathcal{A}_S$,
\begin{align*}&\Pr[\Sim'_{T,q_T}|_S=a_S]\\&=\sum_{a_T: a_T|_S=a_S} \Pr[\Sim'_{T,q_T}=a_T]\\&=\sum_{a_T: a_T|_S=a_S}\max_{q^*\in\mathcal{Q}|q^*_T=q_T}\Pr[\tilde{p}_{q^*}|_T=a_T]\\&\leq\sum_{a_T: a_T|_S=a_S} \E_{q^*\leftarrow \pi|q^*|_T=q_T}\Pr[p_{q^*}|_T=a_T]\\&=\E_{q^*\leftarrow \pi|q^*_T=q_T}\Pr[p_{q^*}|_S=a_S]\\&\leq\Pr[\Sim_{S,q_S}=a_S].\end{align*}By the definition of $\Sim'_{S,q_S}$, $\Sim^{(1)}_{S,q_S}$, and $\tilde{p}_q$, it holds that $$\Pr[\Sim_{S,q_S}=a_S]\geq \Pr[\Sim'_{S,q_S}=a_S]\geq \Pr[\Sim^{(1)}_{S,q_S}=a_S].$$
This, together with the definition of $\xi_{T,q_T}$ implies that
\begin{equation}\label{xi-ST}
\xi_{T,q_T}(S,a_S)\leq \Pr[\Sim_{s,q_S}=a_S]-\Pr[\Sim^{(1)}_{S,q_S}=a_S].
\end{equation}
Moreover, by definition for every $q\in\mathcal{Q}$ and $S\subseteq[k]$
\begin{equation}\label{eqn:sim1-sim'}
\sum_{a_S}\Pr[\Sim^{(1)}_{S,q_S}=a_S]=\sum_{a_S}\Pr[\Sim'_{S,q_S}=a_S]-\xi_{S,q_S}.
\end{equation}
Therefore 
$$
\xi_{T,q_T}(S)\leq 1-\sum_{a_S}\Pr[\Sim'_{S,q_S}=a_S]+\xi_{S,q_S}\leq \nu(q)+\xi_{S,q_S},
$$
where the second inequality follows from Equation~\eqref{eqn:Sim'-bot}.  This  implies that 
\begin{equation}\label{eqn:recurse1}
\xi_{T,q_T}\leq |T|\cdot\nu(q) +\sum_{S\subsetneq T:|S|=|T|-1}\xi_{S,q_S}.
\end{equation}
We use Equation~\eqref{eqn:recurse1}, to prove that for every $T\subseteq[k]$ and for every $q_T$,
\begin{equation}\label{eqn:recurse2}
\xi_{T,q_T}\leq (|T|+1)!\cdot\nu(q)
\end{equation}
We prove Equation~\eqref{eqn:recurse2} by induction on the size of $T$, starting from $|T|=1$.  
For every $T$ of size~$1$ and for every $q_T$, by definition $\xi(T,q_T)=0$.

Suppose Equation~\eqref{eqn:recurse2} holds for every $T$ of size less than~$i$, we prove that it holds for $T$ of size $i$ as follows:
\begin{align*}
\xi_{T,q_T}&\leq {i}\cdot \nu(q) +\sum_{S\subsetneq T: |S|=i-1}\xi_{S,q_S}\\
&\leq
{i}\cdot\nu(q)+i\cdot i!\cdot\nu(q)\\
&\leq(i+1)!\cdot\nu(q)\\
\end{align*}
as desired, where the first inequality follows from Equation~\eqref{eqn:recurse1}, the second inequality follows from the induction hypothesis, and the other inequalities follow from basic arithmetic.

\qed

\medskip

\paragraph{Step 2.} Convert  $\{\Sim^{(1)}_{S,q_S}\}$ into a family of non-signaling distributions.
Similarly to Step~1, we carry out this step via two sub-steps.

\paragraph{Step 2(a).}  We first ensure that the probability that 
$\Sim^{(1)}_{S,q_S}$  outputs~$\bot$ is independent of $q_S$.  
To this end, note that by Equation~\eqref{eqn:sum-Sim*}
\begin{equation}\label{eqn:eps1}
 \E_{q\leftarrow\pi}[\Sim^{(1)}_{[k],q}=\bot]\leq 2^{k\log k}\epsilon\triangleq \epsilon_1.
\end{equation}
Consider the set
$$\GOOD = \left\{ q\in\mathcal{Q} |~ \Pr[\Sim^{(1)}_{[k],q}=\bot]\leq \epsilon_1/\delta \right\},$$
where $\delta$ is from the theorem statement.
By Markov's inequality 
\begin{equation}\label{eqn:good}
\Pr_{q \leftarrow \pi}[ q \in \GOOD ] \geq 1-\delta.
\end{equation}

Note that if $\{p_q\}$ is a strategy in $\subNS_\delta(\mathcal{G})$ and if $\delta\leq \sqrt{\epsilon}$ 
then by Equation~\eqref{eqn:*2},  for every $q\in\mathcal{Q}$
$$ \Pr[\Sim^{(1)}_{[k],q}=\bot]\leq (k+2)!\cdot\nu(q)\leq (k+2)!\cdot2^k\cdot \delta\leq 2^{k\log k}\cdot\delta$$
which implies that 
$\GOOD=\mathcal{Q}$.

Consider the distribution $\pi^* = \pi | (q \in \GOOD)$, and let $\mathcal{G}^*=(\mathcal{Q},\mathcal{A},V,\pi^*)$. 
Note that $\{\Sim^{(1)}_{S,q_S}\}$ is a sub-non-signaling strategy for the game  $\mathcal{G}^*$ whose  value is at least $1 - \epsilon_1$.  This follows from the fact that this is true for the game  $\mathcal{G}$ (see Equation~\eqref{eqn:sum-Sim*}) and from the fact that queries $q\notin \GOOD$ only lower the expected probability of acceptance since they are rejected with probability at least $\epsilon_1/\delta$.

In what follows, we define $\Sim^{(2)}_{S,q_S}$, which is a modification of $\Sim^{(1)}_{S,q_S}$, such that for every $\ell\in[k]$ there exists $\alpha_\ell\in[0,1]$ such that for every $S\subseteq[k]$ of size $\ell$ and for every $q\in\GOOD$, it holds that  \begin{equation}\label{eqn:**} 
\sum_{a_S} \Pr[ \Sim^{(2)}_{S,q_S} = a_S ] = \alpha_\ell
\end{equation}

In addition, we still ensure that for every $S,T\subseteq [k]$ such that $S\subseteq T$, and for every $q\in\mathcal{Q}$ and  $a\in\mathcal{A}$,
\begin{equation}\label{eqn:3*}   \Pr[ \Sim^{(2)}_{T,q_T}|_S = a_S] \leq \Pr[ \Sim^{(2)}_{S,q_S} = a_S ] 
\end{equation}
and 
\begin{equation}\label{eqn:v3}
    \Pr_{q\leftarrow\pi^*, a\leftarrow \Sim^{(2)}_{[k],q}} [V(q,a)=1]
    \geq 1-(k+1)\epsilon_1/\delta\triangleq 1-\epsilon_2.
\end{equation}
To this end, for every $S\subseteq[k]$ and every $q\in\GOOD$ let 
\begin{equation}\label{eqn:beta}
\begin{split}
\beta_{q_S} \triangleq \sum_{a_S}\Pr[ \Sim^{(1)}_{S,q_S} = a_S ]\geq \sum_{a_S}\Pr[ \Sim^{(1)}_{[k],q}|_S = a_S ]= \sum_{a}\Pr[ \Sim^{(1)}_{[k],q} = a]\geq 1-\epsilon_1/\delta,
\end{split}
\end{equation} 
where the first inequality follows from Equation~\eqref{eqn:*} and the last inequality follows from the definition of $\GOOD$.
 
 For every $\ell\in[k]$, let 
\begin{equation}
    \label{eqn:alpha}
\alpha\triangleq(1-\epsilon_1/\delta)
\end{equation} 
For every $S\subseteq[k]$ of size~$\ell$, and for every $q_S\in\mathcal{Q}_S$ and $a_S\in\mathcal{A}_S$, define $$\Pr[ \Sim^{(2)}_{S,q_S} = a_S]\triangleq
\Pr[ \Sim^{(1)}_{S,q_S} = a_S]\cdot\frac{\alpha^\ell}{\beta_{q_S}}.
$$
Note that by definition, for every $S\subseteq[k]$ of size~$\ell$ and for every $q_S\in\mathcal{Q}_S$
\begin{equation}\label{eqn:Sim-alpha}
\sum_{a_S\in\mathcal{A}_S} \Pr[ \Sim^{(2)}_{S,q_S} = a_S ] = \alpha^\ell,
\end{equation}
as desired. Moreover,  note that 
$$
\Pr[ \Sim^{(2)}_{S,q_S} = a_S]=
\Pr[ \Sim^{(1)}_{S,q_S} = a_S]\cdot\frac{\alpha^\ell}{\beta_{q_S}}
\leq \Pr[ \Sim^{(1)}_{S,q_S}=a_S], 
$$
where the first equality follows from the definition of $\Sim^{(2)}_{S,q_S}$ and the last inequality follows from Equations~\eqref{eqn:beta} and~\eqref{eqn:alpha}.   This implies that $$\sum_{a_S\in\mathcal{A}_S}\Pr[\Sim^{(2)}_{S,q_S} = a_S]\leq 1.$$  In the remaining probability $\Sim^{(2)}_{S,q_S}$ outputs~$\bot$.

We next argue that $\Sim^{(2)}_{S,q_S}$ satisfies Equation~\eqref{eqn:3*}.  To this end, fix any $S\subset T\subseteq[k]$ and fix any $q\in\mathcal{Q}$ and $a\in\mathcal{A}$.    Note that 
\begin{align*}  
&\Pr[ \Sim^{(2)}_{T,q_T}|_S = a_S] =\\
&\Pr[ \Sim^{(1)}_{T,q_T}|_S=a_S]\cdot\frac{\alpha^{|T|}}{\beta_{q_T}}\leq\\
&\Pr[ \Sim^{(1)}_{S,q_S} = a_S]\cdot \frac{\alpha^{|T|}}{\beta_{q_T}}\leq \\
&\Pr[ \Sim^{(1)}_{S,q_S} = a_S]\cdot  \frac{\alpha^{|S|}}{\beta_{q_S}}=\\ &\Pr[ \Sim^{(2)}_{S,q_S} = a_S],
\end{align*}
as desired, where the first equation follows from the definition of $\Sim^{(2)}_{T,q_T}$, the second equation follows from Equation~\eqref{eqn:*}, 
the third equation follows from Equations~\eqref{eqn:beta} and~\eqref{eqn:alpha}, and the last equation follows again from the definition of $\Sim^{(2)}_{S,q_S}$.

Finally, note that:
\begin{align*}
&\Pr_{q\leftarrow\pi^*,a\leftarrow \Sim^{(2)}_{[k],q}} [V(q,a)=1]=\\ 
&\Pr_{q\leftarrow\pi^*,a\leftarrow \Sim^{(1)}_{[k],q}} [V(q,a)=1]\cdot\frac{\alpha^k}{\beta_q}\geq \\
&\Pr_{q\leftarrow\pi^*,a\leftarrow \Sim^{(1)}_{[k],q}} [V(q,a)=1]\cdot \alpha^k\geq\\
&(1-\epsilon_1/\delta)^{k+1}\geq 1-(k+1)\epsilon_1/\delta=1-\epsilon_2
\end{align*}
as desired, where the first equation follows from the definition of $\Sim^{(2)}_{[k],q}$, the second equation follows from the fact that $\beta_q\leq 1$, the third equation follows from Equation~\eqref{eqn:sum-Sim*} and from the definition of $\alpha^k$ (Equation~\eqref{eqn:alpha}), the forth equation follows from basic arithmetics, and the last follows by definition of $\epsilon_2$.\\

\paragraph{Step 2(b).} We next define an honest-referee non-signaling strategy for the game $\mathcal{G}^*$ that convinces $V$ to accept with probability at least $1-\epsilon_2$.  More specifically, we define a strategy for which the non-signaling condition holds for every query $q\in\GOOD$.  We note that if $\GOOD=\mathcal{Q}$ (which is the case if $\{p_q\}\in\subNS_\delta(\mathcal{G})$) then the strategy we define is non-signaling.

Our honest-referee non-signaling strategy for the game $\mathcal{G}^*$ is not defined over $\mathcal{A}$ but over $\mathcal{A}^*=\mathcal{A}^*_1\times\ldots\times\mathcal{A}^*_k$, where for each $i\in[k]$, $\mathcal{A}^*_i\triangleq\mathcal{A}_i\cup\{*\}$.
  
We define this strategy in stages. 
First we define a family of distributions $\{\Sim^{(3)}_{S,q_S}\}$ that continues to satisfy the constraints that for every $S,T\subseteq[k]$ such that $S\subseteq T$, and every $q\in\mathcal{Q}$ and $a_S\in\mathcal{A}_S$,
\begin{equation}\label{eqn:Sim3leq}  
\Pr[ \Sim^{(3)}_{T,q_T}|_S = a_S]\leq \Sim^{(3)}_{S,q_S} = a_S ],
\end{equation}
and that there exists constants $\{\alpha_i\}_{i\in[k]}$ such that
\begin{equation}\label{eqn:Sim3total}
\sum_{a_S\in\mathcal{A}_S}\Pr[\Sim^{(3)}_{S,q_S}=a_S]=\alpha_{|S|}.
\end{equation}
At the same time, it also satisfies that for every $q\in\mathcal{Q}$ and every $S\subseteq[k]$,
\begin{equation}\label{eqn:claim-Sim3}
\sum_{i=0}^{k-|S|} (-1)^{i} \sum_{T \supseteq S, |T| = |S| + i}  \Pr[ \Sim^{(3)}_{T,q_T}|_S = a_S ]\geq 0.  
\end{equation}
To this end, we define 
$$
\Pr[ \Sim^{(3)}_{S,q_S} = a_S ] \triangleq \frac{1}{k^{2|S|}} \Pr[ \Sim^{(2)}_{S,q_S} = a_S ].
$$
Note that 
\begin{align*}
&\sum_{i=1}^{k-|S|} (-1)^{i-1} \sum_{T \supseteq S, |T| = |S| + i}  \Pr[ \Sim^{(3)}_{T,q_T}|_S = a_S ]\leq \\
&\sum_{i=1}^{k-|S|} \sum_{T \supseteq S, |T| = |S| + i}  \Pr[ \Sim^{(3)}_{T,q_T}|_S = a_S ]= \\
&\sum_{i=1}^{k-|S|} \sum_{T \supseteq S, |T| = |S| + i}  \frac{1}{k^{2(|S|+i)}}\Pr[ \Sim^{(2)}_{T,q_T}|_S = a_S ]\leq \\
&\sum_{i=1}^{k-|S|} \sum_{T \supseteq S, |T| = |S| + i}  \frac{1}{k^{2(|S|+i)}}\Pr[ \Sim^{(2)}_{S,q_S} = a_S ]\leq \\
&\sum_{i=1}^{k-|S|} {{k-|S|}\choose{i}} \frac{1}{k^{2(|S|+i)}}\Pr[ \Sim^{(2)}_{S,q_S} = a_S ]\leq \\
&\frac{1}{k^{2|S|}}\Pr[ \Sim^{(2)}_{S,q_S} = a_S ]\cdot \sum_{i=0}^{k-|S|} \frac{1}{k^{2i}}{{k-|S|}\choose{i}} \leq \\
&\Pr[ \Sim^{(3)}_{S,q_S} = a_S ]\cdot \sum_{i=0}^{k-|S|} \frac{1}{k^{2i}}{{k-|S|}\choose{i}} \leq\\
&\Pr[ \Sim^{(3)}_{S,q_S} = a_S ],
\end{align*}
where the last inequality follows from the fact that 
\begin{align*}
&\sum_{i=0}^{k-|S|} \frac{1}{k^{2i}}{{k-|S|}\choose{i}} \leq \sum_{i=0}^{k-|S|} \frac{1}{k^{2i}}\cdot k^{i} = \sum_{i=0}^{k-|S|} \frac{1}{k^{i}}\leq\\ &\sum_{i=0}^{k-|S|} 2^{-i}\leq 1.
\end{align*}
We note that the fact that
\begin{align*}
  &\sum_{i=1}^{k-|S|} (-1)^{i-1} \sum_{T \supseteq S, |T| = |S| + i}  \Pr[ \Sim^{(3)}_{T,q_T}|_S = a_S ]\leq \\ 
  &\Pr[ \Sim^{(3)}_{S,q_S} = a_S ] \end{align*}
immediately implies
Equation~\eqref{eqn:claim-Sim3}.

Moreover, by definition of $\Sim^{(3)}_{S,q_S}$ and by Equation~\eqref{eqn:v3},
\begin{equation}\label{eqn:Sim3acc}  
\Pr_{q\leftarrow\pi^*,a\leftarrow \Sim^{(3)}_{k,q}}[V(q,a)=1]\geq \frac{1}{k^{2k}}\cdot(1-\epsilon_2).
\end{equation}
We note that Equation~\eqref{eqn:Sim3total} follows immediately from the definition of $\{\Sim^{(3)}_{S,q_S}\}$ together with Equation~\eqref{eqn:**}.

To argue that Equation~\eqref{eqn:Sim3leq} holds note that for every $S,T\subseteq[k]$ such that $S\subseteq T$, and every $q\in\mathcal{Q}$ and $a_S\in\mathcal{A}_S$,
\begin{align*}
 \Pr[ \Sim^{(3)}_{T,q_T}|_S = a_S] & = \frac{1}{k^{2|T|}}\Pr[\Sim^{(2)}_{T,q_T}|_S = a_S]\\
 &\leq \frac{1}{k^{2|T|}}\Pr[ \Sim^{(2)}_{S,q_S} = a_S ]\\ &\leq\frac{1}{k^{2|S|}}\Pr[ \Sim^{(2)}_{S,q_S} = a_S ] \\
 &=\Pr[ \Sim^{(3)}_{S,q_S} = a_S ].
\end{align*}

Next we define the honest-referee non-signaling strategy $\{p^{**}_q\}$ over $\mathcal{A}^*$.  To this end we first define $\{p^{*}_q\}$, where for every non-empty set $S\subseteq [k]$ and every $a_S\in \mathcal{A}_S$, 
\begin{align*}
    &\Pr[ p^{*}_q= (a_S, (*)^{k-|S|}) ] \triangleq\\ 
    &\sum_{i = 0}^{k - |S|} (-1)^i \sum_{T \supseteq S, |T| = |S| + i} \Pr[ \Sim^{(3)}_{T,q_T}|_S = a_S ],
\end{align*}
 Equation~\eqref{eqn:claim-Sim3} implies that this value is non-negative. Moreover, note that for every $a\in\mathcal{A}$,
 $$\Pr[ p^{*}_q = a ] = \Pr[ \Sim^{(3)}_{[k],q} = a ].
 $$
 We next convert $\{p^{*}_q\}$ into the honest-referee non-signaling strategy $\{p^{**}_q\}$.  To this end, we first define $$\alpha \triangleq \sum_{a\in\mathcal{A}^*}\Pr[p^*_q=a].$$
Note that Equation~\eqref{eqn:Sim3total}, together with the definition of $p^*_q$, implies that $\alpha$ does not depend on~$q$.
In addition, note that
\begin{align*}
&\alpha=\sum_{a\in\mathcal{A}^*}\Pr[p^{*}_q=a]=\\
&\sum_{S\subseteq[k]}\sum_{a_S\in\mathcal{A}_S}\Pr[p^*_q=(a_S,(*)^{k-|S|}]\leq\\
&\sum_{S\subseteq[k]}\sum_{a_S\in\mathcal{A}_S}\sum_{i=0}^{k-|S|}\sum_{T\supseteq S,|T|=|S|+i}\Pr[\Sim^{(3)}_{S,q_S}=a_S]=\\
&\sum_{S\subseteq[k]}\frac{1}{k^{2|S|}}\sum_{a_S\in\mathcal{A}_S}
\sum_{i=0}^{k-|S|}\sum_{T\supseteq S,|T|=|S|+i}\Pr[\Sim^{(2)}_{S,q_S}=a_S]\leq\\
&\sum_{S\subseteq[k]}\frac{1}{k^{2|S|}}
\sum_{i=0}^{k-|S|}\sum_{T\supseteq S,|T|=|S|+i} 1\leq \sum_{S\subseteq[k]}\frac{1}{k^{2|S|}}\cdot 2^{k-|S|}\leq 2^k.\\
\end{align*}
We convert $p^{*}_q$ to a distribution $p^{**}_q$ defined as follows:
If $\alpha\geq 1$ then we convert $p^{*}_q$ to a distribution $p^{**}_q$ defined as follows:  For every $a\in\mathcal{A}^*$, $$\Pr[p^{**}_q=a]\triangleq \frac{1}{\alpha}\Pr[p^{*}_q=a].$$
If $\alpha<1$ then we convert $p^{*}_q$ to a distribution $p^{**}_q$ defined as follows:
$$
\Pr[p^{**}_q=(*)^k]\triangleq 1-\alpha,
$$
and for every $a\in\mathcal{A}^*\setminus \{(*)^k\}$ let 
$$
\Pr[p^{**}_q=a]\triangleq\Pr[p^{*}_q=a].
$$
It is easy to see that $p^{**}_q$ is a distribution.
Moreover,
\begin{align*}
&\Pr_{q\leftarrow\pi^*,a\leftarrow p^{**}_q}[V(q,a)=1]\geq \frac{1}{2^k}\Pr_{q\leftarrow\pi^*,a\leftarrow \Sim^{(3)}_{k,q}}[V(q,a)=1]\geq \frac{1}{k^{3k}}(1-\epsilon_2)\geq \frac{1}{k^{3k}}(1-k^{2k}\epsilon/\delta),
\end{align*}
as desired, where the first inequality follows from the fact that $\alpha\leq 2^k$ together with the definition of~$p_q^{**}$, the second inequality follows from Equation~\eqref{eqn:Sim3acc}, and the third inequality follows from Equations~\eqref{eqn:v3} and~\eqref{eqn:eps1}.
\begin{claim}
$\{p^{**}_q\}$ satisfies the honest referee no-signaling condition.
\end{claim}
\begin{proof}
In what follows, we use the following notation:  If $p^{*}_q$ satisfies  $\alpha=\sum_{a^*\in\mathcal{A}^*}\Pr[p^{*}_q=a]>1$ then let $\gamma=\frac{1}{\alpha}$, and otherwise let $\gamma=1$.

  Fix any subset  $S\subseteq[k]$.   We argue that for every $q,q^*\in\GOOD$ such that $q_S=q^*_S$, and for every $a_S\in\mathcal{A}^*_S$,  
$$\Pr[p^{**}_{q}|_S = a_S]=\Pr[p^{**}_{q^*}|_S = a_S].
$$  
Define $S'\subseteq S$ to be the subset for which for every $i\in S'$ it holds that $a_i\in\mathcal{A}$, and for every $i\in S\setminus S'$ it holds that $a_i=*$.   

\begin{align*}
  &\Pr[ p^{**}_q|_S = a_S ] = \\
  &\sum_{V\subseteq[k]\setminus S}\sum_{a_V\in\mathcal{A}_V}\Pr[p^{**}_q=(a_{S'},a_V,(*)^{k-|S'\cup V|})=\\
  &\sum_{V\subseteq[k]\setminus S}\sum_{a_V\in\mathcal{A}_V}\sum_{i=0}^{k-|S'\cup V|}(-1)^i\sum_{T \supseteq S'\cup V, |T| = |S'\cup V| + i} \gamma\cdot 
  \Pr[ \Sim^{(3)}_{T,q_T}|_{S'\cup V} = a_{S'\cup V} ]=\\ 
   &\sum_{V\subseteq[k]\setminus S}\sum_{i=0}^{k-|S'\cup V|}(-1)^i\sum_{T \supseteq S'\cup V, |T| = |S'\cup V| + i} \gamma\cdot \sum_{a_v\in\mathcal{A}_V}\Pr[ \Sim^{(3)}_{T,q_T}|_{S'\cup V} = a_{S'\cup V} ]=\\ 
    &\sum_{V\subseteq[k]\setminus S}\sum_{i=0}^{k-|S'\cup V|}(-1)^i\sum_{T \supseteq S'\cup V, |T| = |S'\cup V| + i}   
     \gamma\cdot \Pr[ \Sim^{(3)}_{T,q_T}|_{S'} = a_{S'} ] =\\ 
    & \sum_{T\supseteq S'}\gamma\cdot \Pr[ \Sim^{(3)}_{T,q_T}|_{S'} = a_{S'} ]\cdot\left(\sum_{V\subseteq T\setminus S}(-1)^{|T|-|S'\cup V|}\right)
   \end{align*} 
Therefore, to argue that indeed 
$$
\Pr[p^{**}_q|_S=a_S]=\Pr[p^{**}_{q^*}|_S=a_S]
$$
it suffices to prove that for every $T\supseteq S'$ such that $\ell\triangleq |T\setminus S|\geq 1$, it holds that 
$$
\sum_{V\subseteq T\setminus S}(-1)^{|T|-|S'\cup V|}=0,
$$
or equivalently that for every such $T$,
$$
\sum_{V\subseteq T\setminus S}(-1)^{|S'\cup V|}=0.
$$
This follows from the following calculation:

\begin{align*}
&\sum_{V\subseteq T\setminus S}(-1)^{|S'\cup V|}=
    (-1)^{|S'|}\cdot \sum_{V\subseteq T\setminus S}(-1)^{|V|}=\\
    &\sum_{j=0}^{\ell}{{\ell}\choose{j}}(-1)^{j}=(1-1)^\ell=0,
\end{align*}
as desired. 
\end{proof}
\section{Acknowledgements}
We would like to thank Thomas Vidick and Lisa Yang for numerous illuminating and fruitful discussions.  In particular, Thomas was instrumental in formalizing and understanding the notion of $\subNS_\delta$. We would also like to thank the anonymous referee, Thomas Vidick, and Justin Holmgren for their invaluable comments on a previous version of this manuscript. Dhiraj Holden was supported by NSF MACS - CNS-1413920.
\bibliographystyle{plain}
\bibliography{bibfile}
\end{document}